\documentclass[lettersize,journal]{IEEEtran}
\usepackage{cite}
\usepackage{amsmath,amssymb,amsfonts, bm, algorithm, subcaption, adjustbox}
\usepackage{amsthm}
\usepackage{algorithmic}
\usepackage{graphicx}
\usepackage{verbatim}
\usepackage{textcomp}
\usepackage{xcolor}
\usepackage{url}
\usepackage{stfloats}
\usepackage{pifont}
\usepackage{subcaption}
\usepackage{graphicx}
\usepackage{lipsum} 
\def\hlb{\color{blue}}
\def\hlb{\color{black}}
\DeclareMathOperator*{\argmax}{arg\,max}

\newtheorem{theorem}{Theorem}

\def\BibTeX{{\rm B\kern-.05em{\sc i\kern-.025em b}\kern-.08em
		T\kern-.1667em\lower.7ex\hbox{E}\kern-.125emX}}
\begin{document}

	\title{{\hlb When to Use Wireless Challenge-Response Physical Layer Authentication: Design of a Measurable Guideline for OFDM}}
	
	\author{Haiyun Liu,~\IEEEmembership{Student Member,~IEEE,} Shangqing Zhao,~\IEEEmembership{Member,~IEEE,}
		Yao Liu,~\IEEEmembership{Senior Member,~IEEE,}
		and Zhuo Lu,~\IEEEmembership{Senior Member,~IEEE}
		\thanks{Haiyun Liu and Yao Liu are with the Bellini College of Artificial Intelligence, Cybersecurity and Computing,
			University of South Florida, Tampa, FL 33620 USA (e-mail: haiyunliu@usf.edu; yliu21@usf.edu).}
		\thanks{Shangqing Zhao is with the School of Computer Science, University of Oklahoma, OK 73019 USA (e-mail: shangqing@ou.edu).}
		\thanks{Zhuo Lu is with the Department of Electrical Engineering,
			University of South Florida, Tampa, FL 33620 USA (e-mail: zhuolu@usf.edu).}%
		
	}
	
	\maketitle
	
	\begin{abstract}
		\protect {\hlb The security of wireless challenge-response Physical Layer Authentication (PLA) based on  Orthogonal Frequency Division Multiplexing (OFDM) relies on a sufficiently random fading channel condition, which is commonly assumed in existing studies. However, in practical scenarios, such a condition is not always guaranteed and the responses of OFDM subchannels may exhibit correlation.} Consequently, ensuring the security of such PLA systems remains an unsolved problem. In this paper, we propose a novel adversary model, called Maximum Differential Likelihood Generator (MDLG), which exploits the weak correlation property in practical wireless channel to launch effective attacks against PLA. Based on this model, we create a measurable guideline using randomness testing to decide when we can in fact use PLA in a practical wireless channel condition. Extensive real-world experiments validate the effectiveness of the MDLG attack and demonstrate how the proposed guideline can help protect the security of PLA.
	\end{abstract}
	

	\section{Introduction}
	Physical Layer Authentication (PLA) is a widely adopted technique for verifying transmitter identities in wireless systems \cite{bloch2011physical}. Compared with conventional cryptographic approaches, PLA verifies device legitimacy by leveraging inherent physical-layer features or wireless channel characteristics \cite{zeng2010non}, offering a lightweight alternative that is better suited for resource-constrained wireless systems such as Internet of Things (IoT) and Radio Frequency Identification (RFID) systems \cite{gao2020physical, das2009two, wang2018towards}. In a commonly used PLA scheme \cite{shan2013phy, zeng2010non, du2014physical, choi2018coding, yang2013transmit, hassanieh2015securing, fang2018learning, wu2014physical}, two legitimate users, Alice and Bob, securely verify each other’s identity by combining a shared secret key with the wireless channel response between them. An adversary, Eve, would face a significant challenge in attempting to recover the secret key, as she cannot access the random channel between Alice and Bob.
	
	The security of wireless PLA relies on sufficiently random fading conditions. As a result, many existing studies \cite{xiao2007finger,xiao2008time,gao2023EsaNet} have pointed out that the wireless PLA designs should be used in rich scattering environments. However, this is not a well-defined standard. A practical wireless fading environment may not be random as assumed in a design, which offers a potential opportunity for the adversary Eve to crack the secret key or compromise the PLA between Alice and Bob. A clear, measurable guideline has yet been established to answer when or under what conditions we can securely use wireless PLA. 
	
	In this paper, we aim to fill a critical gap between existing research efforts for wireless PLA and the practical issue of its usability: when should Alice and Bob use the PLA? When the wireless subchannels is not quite varying, a PLA design presuming that the channel is sufficiently random and hiding the key behind the channel randomness would fail. But how can the attacker Eve compromise the design under such a condition? It is necessary to design a formal adversary model for Eve that leverages the correlated channel responses to launch an attack. We propose a novel attack model called the Maximum Differential Likelihood Generator (MDLG), which leverages the correlation in wireless channel responses to guess what Alice and Bob's secret key would be given the observation of their signals received at Eve. We formally derive the attack success probability of the MDLG attack and define the security strength that a PLA design can achieve under a wireless channel condition. 
	
	Under the MDLG attack model, we should not use wireless PLA when the wireless channel is not sufficiently random to achieve a target security strength. Therefore, the question of when to use such an PLA design depends on whether the wireless channel quality meets the target security strength. This means we need to evaluate the channel to determine when it is random enough to support PLA. As a result, we propose to use randomness tests \cite{rukhin2001statistical} to evaluate the channel responses and create a design guideline of rejecting the use of PLA when the predicted attack success probability under MDLG is higher than a target security strength (i.e., the maximally allowable attack success probability for a design).

	To the best of our knowledge, the proposed MDLG model is the first adversary model for formal security analysis in wireless PLA. The proposed design guideline leverages randomness testing to determine when the wireless channel can be securely utilized for PLA and when it cannot. Our work offers important and complementary insights in contrast to existing studies \cite{wu2016artificial,fang2018learning,xiao2008using,xie2021physical,xiao2017phy,shan2013phy,zhang2020physical,wang2022csi,lu2021improved} that primarily concentrate on designing technical procedures of PLA.
	
	We conduct extensive real-world experiments to show that our guideline can provide the security guarantee in PLA by using commodity WiFi devices Atheros AR5822/AR9580 chipsets and TP-Link WDR4300 AP. The experimental results demonstrate that (i) MDLG is a powerful attack with high attack success probabilities against PLA in practical scenarios, and (ii) by using our design guideline, PLA achieves a high-level security strength under different (even defective) channel conditions. In summary, the main contributions of this paper are as follows.
	
	\begin{enumerate}
		\item We introduce a formal adversary model MDLG to formalize the security analysis and evaluation of PLA under different wireless channels.
		\item We propose a new design guideline to outline when the random channel can indeed support a target security strength goal and how to maximize the efficiency of using PLA given a security requirement.
		\item We use extensive experiments in real-world environments to show our design guideline guarantees the security strength of PLA.
	\end{enumerate}
	
	The rest of this paper is organized as follows: Section~\ref{Sec:background} introduces the background of PLA and our design motivation. Section~\ref{Sec:Adversary} presents the detailed MDLG adversary strategy. Section~\ref{Sec:Defense Method} proposes the defense guideline for security guarantee. Section~\ref{Sec:experiment} shows the results of real-world experiments. Section~\ref{Sec:Relatework} summarizes related works, followed by the conclusion in Section~\ref{Sec:Conclusion}.
	
	\section{Background, Modeling and Motivation}\label{Sec:background}
	In this section, we first introduce the background of wireless key-based challenge-response PLA between two users Alice and Bob. Then, we explain why this PLA scheme can ensure security. Finally, we outline a potential vulnerability and introduce our design motivation.

	\subsection{Basics in Physical Layer Authentication}
	We consider a typical Orthogonal Frequency Division Multiplexing (OFDM) communication system between Alice and Bob as OFDM has been widely adopted in today's wireless networking and is commonly used in existing PLA designs \cite{Melki2019Survey,hou2012CFO,Chin2015Fading,perazzone2021artificial}. Assume that Alice and Bob use $L$ OFDM subcarriers (or subchannels) for communication and they pre-share a binary sequence unknown to others as their secret key $\bm{s}$ that consists of $S$ bits.
	
	Fig.~\ref{Fig:basickeybased} shows a basic outline of the challenge-response PLA scheme. The PLA process involves four steps \cite{wu2014physical, shan2013phy, zeng2010non, du2014physical, choi2018coding, wu2016artificial}: 1) Alice sends a random challenge signal $S_A$ to start the PLA; 2) $S_A$ goes through the channel and Bob receives it as $R_B$; 3) Bob generates his response signal $S_B$ based on $R_B$ and the shared key $\bm{s}$; 4) the response signal $S_B$ travels through the channel and is received by Alice as $R_A$, based on which Alice verifies Bob.
	
	Specifically, to initiate the PLA, Alice first sends Bob a challenge signal \cite{wu2016artificial}, which can be represented as a frequency domain vector as
	\begin{align}
		S_A=[e^{j\beta_0}, ...,e^{j\beta_{L-1}}], 
	\end{align}
	where  $\beta_l$ is a random transmit phase on the $l$-th subcarrier. This signal will go through the wireless channel with responses denoted as
	\begin{align}\label{Eq:AliBobChannel}
		\bm h = [a_0 e^{j \theta_0}, \ldots, a_{L-1} e^{j \theta_{L-1}}],
	\end{align}
	where $a_l$ is the amplitude response and $\theta_l$ is the phase shift on subcarrier~$l$. Then, Bob's received signal is 
	\begin{equation}\label{Eq:ReceivedSignalBob}
		\begin{split}
			R_B=\left[a_0e^{j(\theta_0+\beta_0)},...,a_{L-1}e^{j(\theta_{L-1}+\beta_{L-1})}\right].
		\end{split}
	\end{equation}
	To send a response signal, Bob first uniformly separates the secret key $\bm{s}$ into $L$ sub-keys denoted as $s_0, \dots, s_{L-1}$, then uses a one-one mapping function $M$ to map each sub-key $s_l$ to phase $\phi_l$ ($l\in\{0,...,L-1\}$); i.e., 
	\begin{equation}\label{Eq:MappingM}
		\phi_l = M(s_l) \in \bm{\Phi}.
	\end{equation}
	where $\bm{\Phi}$ is the set of all possible mapped phases (e.g., $\bm{\Phi} = \{0,\pi\}$ in a binary mapping with bits 0 and 1 mapped to phases 0 and $\pi$, respectively). All the $\phi_l$ values form a key-mapped phase sequence, denoted as 
	\begin{equation}\label{Eq:ppp}
		\begin{split}
			\bm \phi = [\phi_0, \ldots, \phi_{L-1}].
		\end{split}
	\end{equation}
	Next, Bob negates the phases of $R_B$ in \eqref{Eq:ReceivedSignalBob} and adds phase sequence $\bm\phi$ to obtain the response signal $S_B$ as
	\begin{equation}\label{Eq:ChallengeSignal}
		\begin{split}
			S_B=\left[a_0e^{j\left(\phi_0-(\theta_0+\beta_0)\right)},...,a_{L-1}e^{j\left(\phi_{L-1}-(\theta_{L-1}+\beta_{L-1})\right)}\right],
		\end{split}
	\end{equation}
	which is then transmitted by Bob to Alice. Since the duration of the challenge-response procedure is typically within the channel coherence time, the wireless channel can be considered time-invariant during this period. Under such conditions, channel reciprocity holds, and Alice’s received signal is
	\begin{equation}\label{Eq:Alicereceive}
		\begin{split}
			R_A=\left[{a_0}^2e^{j(\phi_0-\beta_0)},...,{a_{L-1}}^2e^{j(\phi_{L-1}-\beta_{L-1})}\right].
		\end{split}
	\end{equation}
	Since the phases of $S_A$ (i.e., $[\beta_0, \ldots, \beta_{L-1}]$) and the key-mapped phases (i.e., $[\phi_0, \ldots, \phi_{L-1}]$) are all known to Alice, she can pre-compute an expected phase of the response signal received on subcarrier~$l$ as $(\phi_l - \beta_l)$. Then, she compares the expected phase with the actual phase received in \eqref{Eq:Alicereceive} on each subcarrier. If all phases match, Alice can verify that the response signal was indeed sent by the legitimate user Bob.

	\begin{figure}[t!] 
		\centering
		\includegraphics[width=0.43\textwidth]{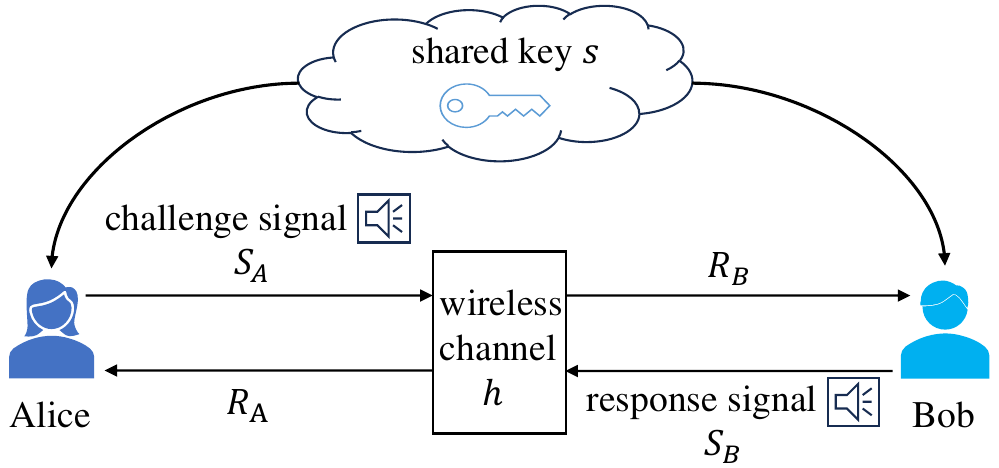}
		\caption{Challenge-response PLA scheme.}
		\label{Fig:basickeybased}
	\end{figure}
	
	\subsection{Assumptions and Modeling for Attacker Eve}
	We consider that the eavesdropper Eve can observe the entire PLA process between Alice and Bob, as illustrated in Fig.~\ref{Fig:Eveperspective}.
	We assume that Eve has full knowledge of their communication setup (e.g., number of subcarriers, carrier frequencies, communication bandwidth, and the key-phase mapping function $M$), and is capable of accurately estimating the channel responses between herself and the legitimate users (i.e., Alice or Bob). Given a possible key candidate $\bm s'$, Alice can verify whether $\bm s'$ is indeed the secret key used by Alice and Bob. However, Eve does not have access to the channel response between Alice and Bob. Eve's objective is to obtain Alice and Bob's key $\bm s$ based on Alice's challenge signal $S_A$ and Bob's response signal $S_B$. 
	
	\begin{figure}[t!]
		\centering
		\includegraphics[width=0.41\textwidth]{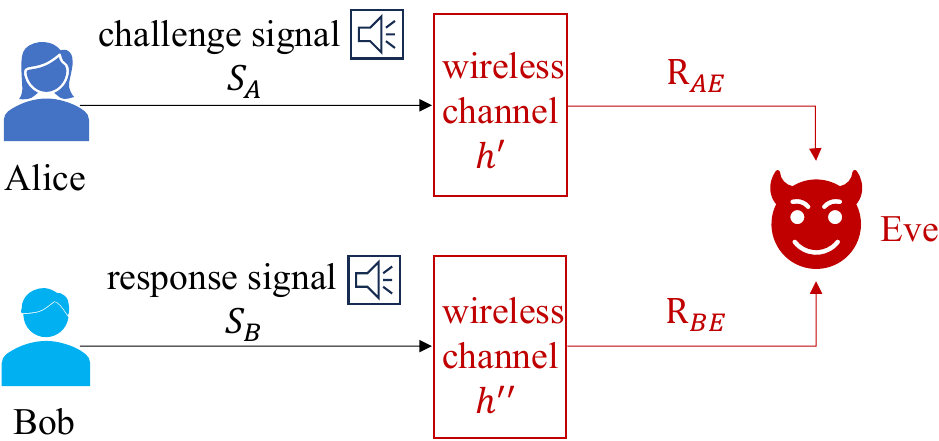}
		\caption{What Eve can observe from Alice and Bob.}
		\label{Fig:Eveperspective}
	\end{figure}
	
	When Alice sends the challenge signal $S_A$ to Bob, Eve can also receive this as $R_{AE}$. Denote channel response between Alice and Eve as
	\begin{equation}
		\begin{split}
			h^{'}=[a_{0}^{'}e^{j\theta_{0}^{'}} ,...,a_{L-1}^{'} e^{j\theta_{L-1}^{'}}],
		\end{split}
	\end{equation}
	where $a_l^{'}$ is the amplitude response and $e^{j \theta_l^{'}}$ is the phase response on subcarrier $l$. Then, the signal received by Eve from Alice is
	\begin{equation}\label{Eq:RAE}
		\begin{split}
			R_{AE}=\left[a_0^{'}e^{j(\theta_0^{'}+\beta_0)},...,a_{L-1}^{'}e^{j(\theta_{L-1}^{'}+\beta_{L-1})}\right].
		\end{split}
	\end{equation}
	Because $[\theta_{0}^{'},...,\theta_{L-1}^{'}]$ is known to Eve, she can extract  $[\beta_0, \ldots, \beta_{L-1}]$ from her received signal $R_{AE}$.
	Similarly, denote the channel response between Bob and Eve as
	\begin{equation}
		\begin{split}
			h^{''}=[a_{0}^{''}e^{j\theta_{0}^{''}} ,...,a_{L-1}^{''} e^{j\theta_{L-1}^{''}}],
		\end{split}
	\end{equation}
	where $a_l^{''}$ is the amplitude response and $e^{j \theta_l^{''}}$ is the phase response on subcarrier $l$. Bob's response signal at Eve $R_{BE}$ can be represented as
	\begin{equation}\label{Eq:RBE} 
		\begin{split}
			R_{BE}=\Big[&a_0a_0^{''}e^{j\left(\phi_{0}-(\theta_{0}+\beta_0)+\theta_{0}^{''}\right)},\\
			&...,a_{L-1}a_{L-1}^{''}e^{j\left(\phi_{L-1}-(\theta_{L-1}+\beta_{L-1})+\theta_{L-1}^{''}\right)}\Big].
		\end{split}
	\end{equation}
	Since $[\theta_{0}^{''},...,\theta_{L-1}^{''}]$ and $[\beta_0, \ldots, \beta_{L-1}]$ are known to Eve, the best phase information she can get by comparing the phases between \eqref{Eq:RAE} and \eqref{Eq:RBE} is 
	\begin{equation}\label{Eq:z}
		\begin{split}
			\bm z= [z_1, \ldots, z_{L-1}] = [\phi_0 - \theta_0, \ldots, \phi_{L-1} - \theta_{L-1}].
		\end{split}
	\end{equation}
	As Eve lacks access to the channel between Alice and Bob, she has no knowledge of $[\theta_0, \ldots, \theta_{L-1}]$. Most existing studies \cite{Ze2010,xiaowen2003adaptive,yaghoobi2004scalable,tsao2014performance,mabrouk2012effect} implicitly assume that $\theta_0$, $\ldots$, $\theta_{L-1}$ are independently random due to scatter-rich wireless environments. Under such an assumption, there is no way for Eve to obtain Alice and Bob's key-mapped $\bm \phi=[\phi_0, \ldots, \phi_{L-1}]$ by only obtaining $\bm z$ in \eqref{Eq:z}. 
	
	\subsection{Potential Vulnerability and Attack Design Motivation}
	It is worth noting that the assumption that the channel responses in rich-scattered environments are independent is commonly used for analysis and evaluation of communication performance, rather than security modeling \cite{goldsmith2005wireless}. In a real-world scenario, adjacent OFDM subchannel responses can show a correlation property, especially when the subcarrier spacing is smaller than the coherence bandwidth \cite{hwang2008ofdm}. {\hlb This is because practical wireless environments inherently involve multipath propagation, which leads to a time-domain delay spread and hence a corresponding coherence bandwidth. Meanwhile, to mitigate inter-symbol interference (ISI), modern OFDM systems are typically designed with a subcarrier spacing much smaller than the coherence bandwidth. Consequently, adjacent subcarriers often lie within the same coherence bandwidth and thus experience highly similar amplitude and phase fading.} To illustrate this correlation intuitively, we conducted an indoor OFDM experiment with a 20MHz bandwidth and 56 subcarriers. As shown in Fig.~\ref{Fig:ExamplePhase}, for the four different signals, the phase responses of adjacent OFDM subchannels do not differ significantly and the overall curves are relatively smooth. Specifically, the phase difference between two adjacent subcarriers is approximately $\frac{\pi}{8}$, indicating a non-negligible correlation between their channel responses.
	By calculating the correlation between adjacent subchannel responses and summarizing the results in Fig.~\ref{Fig:ExampleCoefficient}, it can be seen that the distribution of the correlation coefficients is similar to a normal distribution, with a mean of approximately $0.35$.
	\begin{figure}[t!]
		\centering
		\begin{minipage}{0.49\columnwidth}
			\centering		
			\includegraphics[width=1.04\textwidth]{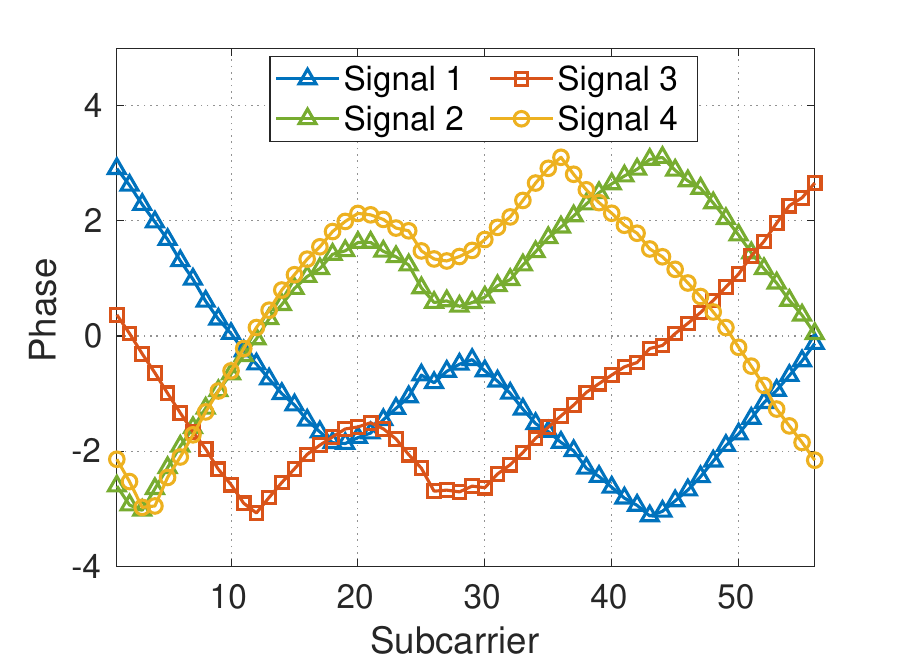}
			\caption{Phase response of subcarriers.}
			\label{Fig:ExamplePhase}
		\end{minipage}
		\hfill
		\begin{minipage}{0.49\columnwidth}
			\centering
			\raisebox{0.02\height}{\includegraphics[width=\textwidth]{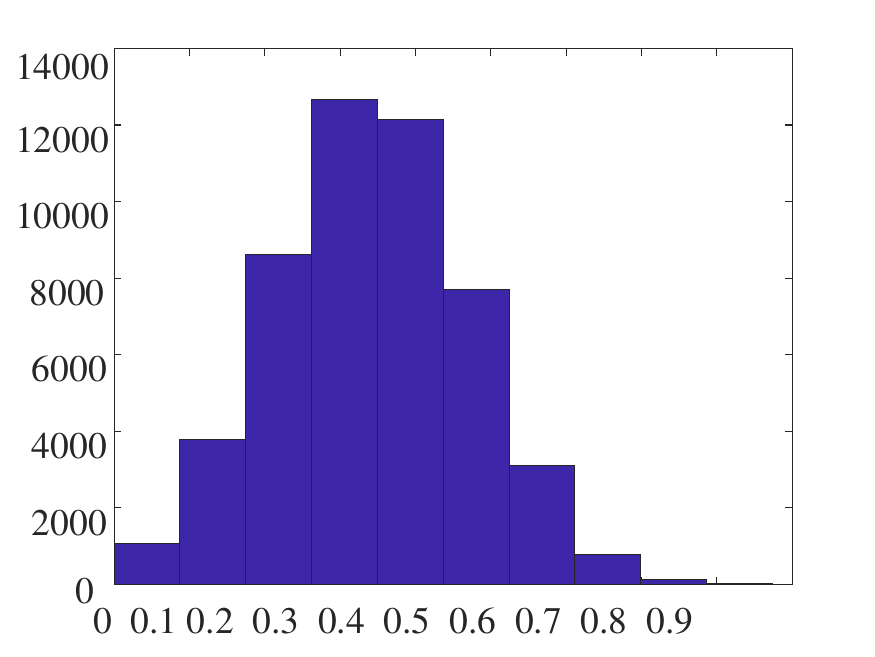}}
			\caption{Correlation coefficient between channel responses.}
			\label{Fig:ExampleCoefficient}
		\end{minipage}
	\end{figure}
	
	Due to the weak correlation property observed in practice, we aim to investigate how Eve can extract useful information about the secret key $\bm s$ from the adjacent subcarrier responses during Alice and Bob's PLA process.

	\section{Adversary Strategy: The MDLG Attack}\label{Sec:Adversary}
	In this section, we will introduce a new adversary model for Eve to attack the PLA (i.e., aiming to infer Alice and Bob's secret key $\bm s$ based on Eve's observed signals). First, we explain how Eve can use her observations to create a potential attack. Then, we propose the MDLG attack strategy and detail the design. Finally, we formalize the attack model and analyze its success probability.

	\subsection{Leverages Eve's Observations Under Ideal Scenario} \label{Sec:ideal}
	
	From \eqref{Eq:ppp} and \eqref{Eq:z}, we know that the key-mapped phase sequence is $\bm \phi = [\phi_0, \ldots, \phi_{L-1}]$ and Eve's best information inferred from Alice's challenge signal \eqref{Eq:RAE} and Bob's response signal \eqref{Eq:RBE} is $\bm z=[\phi_0 - \theta_0, \ldots, \phi_{L-1} - \theta_{L-1}]$ in \eqref{Eq:z}. We let Eve look at the difference between adjacent elements in $\bm z$. In particular, for any $l\in\{0,...,L-2\}$, 
	\begin{equation}\label{Eq:difference}
		\begin{split}
			(z_{l+1}-z_{l}) = (\phi_{{l+1}} - \phi_{{l}}) - (\theta_{{l+1}} - \theta_{{l}}).
		\end{split}
	\end{equation}
	We first consider an ideal scenario for Eve, where Alice's and Bob's OFDM signals both go through flat-fading, noise-free environments \cite{ma2005novel} (i.e., the correlation coefficient between any pair of OFDM subchannels is 1). In this case, it is clear from \eqref{Eq:difference} that $(\theta_{{l+1}} - \theta_{{l}})=0$ and Eve can further obtain
	\begin{equation}\label{Eq:DiffPhi}
		\phi_{{l+1}} - \phi_{{l}} = z_{l+1}-z_{l}.
	\end{equation}
	This means Eve can know the difference between any pair of adjacent elements in $\bm \phi$, i.e., $(\phi_{{1}} - \phi_{{0}})$, ..., $(\phi_{{L-1}} - \phi_{{L-2}})$. If $\phi_0$ is determined, Eve can in turn obtain the values of $\phi_1$, $\phi_2$, ..., $\phi_{L-1}$ one by one as follows:
	\begin{equation}\label{Eq:Getphi}
		\begin{split}
			&	\phi_1=\phi_0+(\phi_{{1}} - \phi_{{0}}),\\ 
			&	\phi_2=\phi_1+(\phi_{{2}} - \phi_{{1}}),\\ 
			&~~~~~~~~~~\vdots\\
			&	\phi_{L-1}=\phi_{L-2}+(\phi_{{L-1}} - \phi_{{L-2}}).
		\end{split}
	\end{equation}
	Then, Eve can reconstruct $\bm{\phi}$, and by reversing the key-phase mapping function $M$ in \eqref{Eq:MappingM}, she can eventually determine Alice and Bob's shared secret key $\bm{s}$.

	More formally, define two differential phase sequences
	\begin{equation}\label{Eq:Deltaz}
		\begin{split}
			\bm\Delta^{\bm z}=[(z_1-z_0), ..., (z_{L-1}-z_{L-2})]
		\end{split}
	\end{equation}
	and 
	\begin{equation}\label{Eq:DeltaPhi}
		\begin{split}
			\bm\Delta^{\bm \phi}=[(\phi_{1}-\phi_{0}), ..., (\phi_{L-1}-\phi_{L-2})].
		\end{split}
	\end{equation}
	In this ideal flat fading example, Eve can obtain the differential key-mapped phase sequence $\bm\Delta^{\bm \phi}$ by directly letting $\bm\Delta^{\bm \phi} = \bm\Delta^{\bm z}$. If $\phi_0$ is determined, Eve can obtain $\bm \phi$ based on $\eqref{Eq:Getphi}$. Subsequently, Eve reverses the mapping $M$ for $\bm \phi$ to obtain a potential key sequence ${\bm s'}$ as
	
	\begin{equation}\label{Eq:ReverseMapping}
		{\bm s'} = [M^{-1}(\phi_0), ..., M^{-1}(\phi_{L-1})].
	\end{equation}

	As $\phi_0 \in \bm\Phi$ in $\eqref{Eq:MappingM}$, which consists of a limited number of  possible phases, Eve can enumerate each possible phase for $\phi_0$, obtain a corresponding key candidate based on $\eqref{Eq:Getphi}$ and \eqref{Eq:ReverseMapping}, then verify whether the candidate is the true key ${\bm s}$.

	\subsection{Attack Design Basics Under Realistic Scenarios}
	Now, we present the attack design under real-world conditions. In practice, adjacent OFDM subchannel responses are typically correlated, and the correlation coefficient does not reach $1$ (e.g., with a mean value of $0.35$, as shown in Fig.~\ref{Fig:ExampleCoefficient}). This implies that term $(\theta_{{l+1}} - \theta_{{l}})$ in \eqref{Eq:difference} is nonzero and \eqref{Eq:DiffPhi} does not hold. However, based on \eqref{Eq:difference}, $(\theta_{l+1} - \theta_{l})$ can be regarded as an offset added to the right-hand side of \eqref{Eq:DiffPhi}. This suggests that $\bm{\Delta}^{\bm{z}}$ can still serve as a reference for estimating $\bm{\Delta}^{\bm{\phi}}$, particularly when the correlation between adjacent subcarriers is strong. This offset decreases as the correlation increases, and approaches zero when the correlation coefficient reaches 1. In this case, the scenario becomes equivalent to the ideal case described in Section~\ref{Sec:ideal}.

	As shown in \eqref{Eq:ReverseMapping}, finding the true key ${\bm s}$ involves converting a phase sequence into a potential key bit sequence determined by the mapping function $M$ in \eqref{Eq:MappingM}. We first consider $M$ as a binary mapping function, which is commonly used in existing studies \cite{jeon2012dual,wu2016artificial,21qzx-tifs}, to describe the intuitive basic attack design. Without loss of generality, assume that $M$ maps bits 0 and 1 to phases 0 and $\pi$, respectively. Under the binary mapping, the bit length $S$ of secret key $\bm s$ is equal to the number of subcarriers $L$.
	
	Based on \eqref{Eq:difference}, each phase in $\bm{\Delta}^{\bm{z}}$ generally does not take a value of exactly $0$ or $\pi$, as is the case for $\bm{\Delta}^{\bm{\phi}}$. To better exploit the relationship between $\bm{\Delta}^{\bm{z}}$ and $\bm{\Delta}^{\bm{\phi}}$, Eve first quantizes each phase $\Delta_l^{\bm{z}} \in \bm{\Delta}^{\bm{z}}$ in \eqref{Eq:Deltaz} into a bit $b^{\bm{z}}_l$ for $l \in \{0, 1, \ldots, L - 2\}$ as
	\begin{equation}\label{Eq:quantization}
		b^{\bm z}_{l}= 
		\begin{cases} 
			0, & \text{if~ } \Delta_{l} \in (-\frac{\pi}{2}, \frac{\pi}{2}] \\
			1, & \text{if~ } \Delta_{l} \in (-\pi, -\frac{\pi}{2}] \cup (\frac{\pi}{2}, \pi],
		\end{cases}
	\end{equation}
	where $(-\frac{\pi}{2}, \frac{\pi}{2}]$ and $(-\pi, -\frac{\pi}{2}] \cup (\frac{\pi}{2}, \pi]$ denote the phase regions close to $0$ (bit $\bm 0$ under binary mapping) and $\pi$ (bit $\bm 1$ under binary mapping), respectively. 
	As a result, Eve obtains a differential bit sequence 
	\begin{equation}\label{Eq:bz}
		\bm b^{\bm z} = [b^{\bm z}_0, ..., b^{\bm z}_{L-2}].
	\end{equation}
	Similarly, each phase $\Delta_l^{\bm{\phi}} \in \bm{\Delta}^{\bm{\phi}}$ in \eqref{Eq:DeltaPhi} is quantized into a bit $b^{\bm{\phi}}_l$ for $l \in \{0, 1, \ldots, L - 2\}$ using the quantization rule in \eqref{Eq:quantization}, from which we can obtain differential bit sequence 
	\begin{equation}\label{Eq:bphi}
		\bm b^{\bm \phi} = [b^{\bm \phi}_0, ..., b^{\bm \phi}_{L-2}].
	\end{equation}
	For each element $b^{\bm \phi}_l$, it represents the bitwise difference between two adjacent bits $s_l$ and $s_{l+1}$ in the key $\bm s$ (i.e., $b^{\bm \phi}_l$ is 0 if $s_l$ and $s_{l+1}$ have the same value or 1 otherwise). If Eve knows $\bm b^{\bm \phi}$, there are only two possible key candidates by setting the first bit to be either 0 or 1. Then, Eve can verify both candidates to find the true key $\bm s$. 
	
	As discussed above, since $\bm{\Delta}^{\bm{z}}$ can serve as a reference for estimating $\bm{\Delta}^{\bm{\phi}}$, the corresponding sequence $\bm{b}^{\bm{z}}$ can likewise be used as an estimate of $\bm{b}^{\bm{\phi}}$. In this regard, Eve can set $\bm b^{\bm z}$ as the most likely candidate for $\bm b^{\bm \phi}$ and verify $\bm b^{\bm z}$ first. However, in many cases, $\bm b^{\bm z}$ may not be exactly the same as $\bm b^{\bm \phi}$. Eve needs to select more other candidates for $\bm{b}^{\bm{\phi}}$ to verify. Under weak subcarrier correlation, she can start from changing 1 bit in $\bm b^{\bm z}$ to get more candidates, then change 2 and more bits until the true key $\bm s$ is identified. {\hlb Specifically, for each candidate for $\bm{b}^{\bm{\phi}}$, denoted by $\hat{b}^{\phi}$, Eve would operate as follows to verify its corresponding key candidates: she first initializes the key candidates reconstruction process by assuming the first bit of the secret key to be $\hat{s}_0 \in \{0, 1\}$; then, the subsequent key bits are determined iteratively based on the candidate sequence $\hat{b}^{\phi}$, such that the next key bit $\hat{s}_{l+1}$ is set equal to $\hat{s}_l$ if the candidate bit $\hat{b}_l^{\phi} = 0$ (indicating identical adjacent bits), and $\hat{s}_{l+1} = \bar{\hat{s}}_l$ if $\hat{b}_l^{\phi} = 1$ (indicating a bit flip), where $\bar{\hat{s}}_l$ denotes the bitwise inverse of $\hat{s}_l$; consequently, this procedure yields two key candidates, which Eve then verifies to determine if they match the true secret key $\bm{s}$.}
	
	\begin{figure}[t!]
		\centering
		\includegraphics[width=0.8\linewidth]{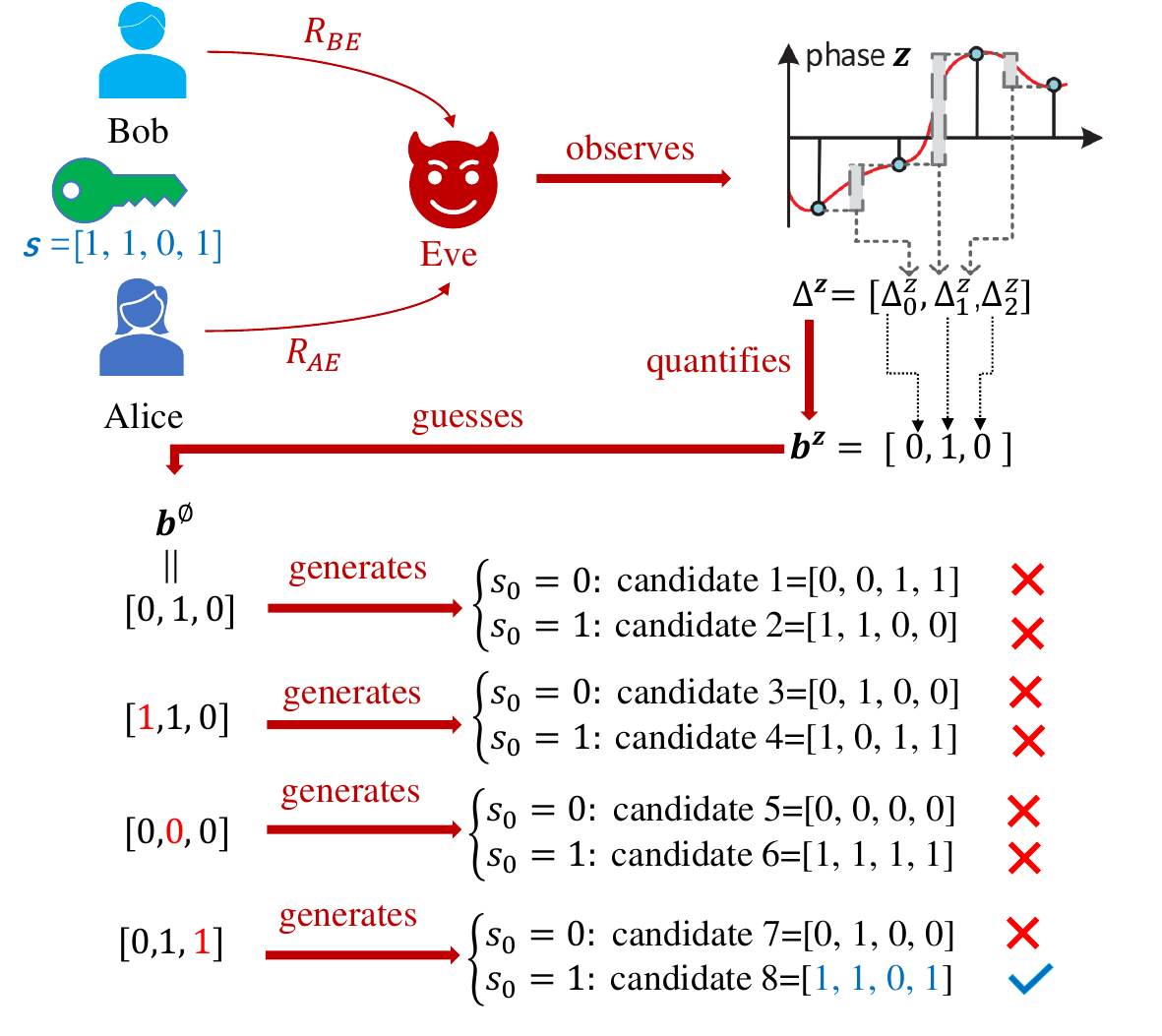}
		\caption{A basic attack design.} 
		\label{Fig:BasicDesign}
	\end{figure}
	
	Fig.~\ref{Fig:BasicDesign} illustrates an example of the basic attack design, where the true secret key is $\bm{s} = [1, 1, 0, 1]$. In this example, Eve first obtains her observation $\bm{z}$, from which she derives the differential phase sequence $\bm{\Delta}^{\bm{z}}$, and then quantizes $\bm{\Delta}^{\bm{z}}$ into a differential bit sequence $\bm{b}^{\bm{z}} = [0, 1, 0]$. Eve uses $\bm{b}^{\bm{z}} = [0, 1, 0]$ as the first candidate for $\bm{b}^{\bm{\phi}}$. She sets the first bit of the key $\bm{s}$ to 0 and 1, and generates two key candidates: $[0, 0, 1, 1]$ and $[1, 1, 0, 0]$, but fails to match the actual key. Then, Eve changes $1$ bit in $\bm b^{\bm z} = [0, 1, 0]$, generates new key candidates and continues to verify. After multiple attempts, when Eve flips the last bit in $\bm{b}^{\bm{z}} = [0, 1, 0]$ to obtain $[0, 1, 1]$ as a candidate for $\bm{b}^{\bm{\phi}}$ and sets the first bit of the key $\bm{s}$ to 1, she eventually verifies that the key candidate $[1, 1, 0, 1]$ is the true key.

	\subsection{The MDLG Strategy}
	As shown in the example in Fig.~\ref{Fig:BasicDesign}, the key idea of Eve's attack strategy is to use $\bm b^{\bm z}$ as a reference to guess $\bm{b}^{\bm{\phi}}$. This guessing is conducted in descending order of likelihood to select candidates for $\bm{b}^{\bm{\phi}}$: starting from the most likely candidate (i.e., the reference $\bm b^{\bm z}$) and moving towards less likely ones by gradually changing more bits in the reference $\bm b^{\bm z}$. We call this strategy maximum differential likelihood generator (MDLG). {\hlb It is worth noting that $\bm b^{\bm z}$ and $\bm{b}^{\bm{\phi}}$ depend only on the relative phase differences between adjacent OFDM subchannels of $\bm z$ and $\bm \phi$ (i.e., $\bm\Delta^{\bm z}$ and $\bm\Delta^{\bm \phi}$), and are independent of the absolute phase values of $\bm z$ and $\bm \phi$; therefore, adopting different quantization thresholds for the key mapping function $M(\cdot)$ will not affect the attack performance.} In the following, we detail the working steps of MDLG and mathematically analyze its attack performance.
	
	\subsubsection{Strategy Design}
	The MDLG strategy consists of two major steps. 
	
	{\noindent\bf{Step 1: Computing the differential sequence $\bm{b}^{\bm{z}}$.}}
	Eve first extracts the phase sequence $\bm{z}$ from \eqref{Eq:z} by comparing the phases in the observations \eqref{Eq:RAE} and \eqref{Eq:RBE}. She then computes the differential phase sequence $\bm{\Delta}^{\bm{z}}$ using \eqref{Eq:Deltaz}, and quantizes $\bm{\Delta}^{\bm{z}}$ based on \eqref{Eq:quantization} to obtain the differential bit sequence $\bm{b}^{\bm{z}}$.
	
	{\noindent\bf{Step 2: Selecting and verifying candidates by gradually changing bits in $\bm b^{\bm z}$.}}
	Eve uses $\bm b^{\bm z}$ as the first candidate for $\bm b^{\bm \phi}$ and generates two corresponding key candidates by setting the first bit in key $\bm s$ to $0$ and $1$, respectively. Eve verifies both candidates to see if she finds the true key $\bm s$. If Eve cannot, she changes $1$ bit in $\bm b^{\bm z}$, enumerates and verifies all possible key candidates under the 1-bit change. If the true key still cannot be found, Eve proceeds to change $2$ or more bits for further verification, until the true key is identified.

	\subsubsection{Performance Analysis of MDLG}
		It is clear that in practice, Eve has a limited computational capability and cannot enumerate all possible key candidates to verify within a reasonable time period. For example, given a key size of 128 bits, it is computationally infeasible for Eve to verify $2^{128}$ possibilities based on today's computing power. As a result, we assume that Eve can only verify up to $N$ possibilities starting from its reference $\bm b^{\bm z}$, where $N$ is called Eve's computing capability. 
		\begin{theorem}\label{The:MDLG}
			Given Eve's capability $N$ and the correlation coefficient among OFDM subchannels $\rho$, the success probability of MDLG can be expressed as
			\begin{equation}\label{Eq: MDLGforRho}
				P_{\text{MDLG}} 
				= {I_{\left(\frac{1-\rho}{2}\right)}\left(L-1-n_{\max}, n_{\max}+1\right)},
			\end{equation}
			where $n_{\max}=\max\{n' | 2\sum_{n=0}^{n'}\binom{L-1}{n} \leq N\}$, and $I_x(a,b) = \frac{B(x;a,b)}{B(1;a,b)}$ represents the regularized incomplete beta function with incomplete beta function $B(x;a,b) = \int_0^x t^{a-1}(1-t)^{b-1} \, dt$ and complete beta function $B(a,b) = B(1;a,b)$.
		\end{theorem}
		
		\noindent\textit{Proof:} Based on MDLG, Eve changes one or more bits in $\bm{b}^{\bm{z}}$ to generate candidates for $\bm{b}^{\bm{\phi}}$ and subsequently verifies the corresponding candidates for $\bm{s}$. We use $\bm b^{\bm{z \phi}}$ to represent the difference between $\bm{b}^{\bm{z}}$ and $\bm{b}^{\bm{\phi}}$, i.e., $\bm b^{\bm{z\phi}}=\bm{b}^{\bm{\phi}} \oplus \bm{b}^{\bm{z}}$ where $\oplus$ is the exclusive OR (XOR) operation. Then, the candidate for $\bm{b}^{\bm{\phi}}$ can be obtained by performing XOR between the candidate for $\bm b^{\bm{z \phi}}$ and the given $\bm{b}^{\bm{z}}$. For Eve, this MDLG strategy is equivalent to increasing the value of $n$ gradually to find the true $\bm s$, where $n$ is the number of bits being $\bm 1$ in the candidate for $\bm b^{\bm{z \phi}}$. Let $n_{\max}$ denote the maximum value that $n$ can reach under Eve's capability $N$; that is, the number of bits equals to $\bm 1$ in the candidate for $\bm b^{\bm{z \phi}}$ is at most $n_{\max}$, or Eve can change at most $n_{\max}$ bits in $\bm{b}^{\bm{z}}$ to get the candidate for $\bm{b}^{\bm{\phi}}$. 
		If the bit difference between $\bm{b}^{\bm{z}}$ and $\bm{b}^{\bm{\phi}}$ is no greater than $n_{\max}$, i.e., the number of bits equal to $\bm{1}$ in $\bm{b}^{\bm{z\phi}}$ does not exceed $n_{\max}$, Eve can succeed in finding the true key $\bm s$ with the capability $N$.
		As a result, we proceed to analyze the statistical properties of the bits in $\bm{b}^{\bm{z\phi}}$. 
		We define another differential sequence, similar to $\bm{\Delta}^{\bm{z}}$ in \eqref{Eq:Deltaz} and $\bm{\Delta}^{\bm{\phi}}$ in \eqref{Eq:DeltaPhi}, as
		\begin{equation}\label{Eq:Deltatheta}
			\begin{split}
				\bm\Delta^{\bm\theta}=[(\theta_1-\theta_0), ..., (\theta_{L-1}-\theta_{L-2})].
			\end{split}
		\end{equation}
		Given that the analysis targets the correlation between adjacent subchannels, the elements in $\bm{\Delta}^{\bm{\theta}}$ should be independently and identically distributed (i.i.d.). Based on \eqref{Eq:Getphi}, the differences between the corresponding elements of $\bm{\Delta}^{\bm{z}}$ and $\bm{\Delta}^{\bm{\phi}}$ should also follow an i.i.d. distribution. Since $\bm{b}^{\bm{z}}$ and $\bm{b}^{\bm{\phi}}$ are generated from $\bm{\Delta}^{\bm{z}}$ and $\bm{\Delta}^{\bm{\phi}}$, respectively, according to the quantization rule in \eqref{Eq:quantization}, the differences between the corresponding bits in $\bm{b}^{\bm{z}}$ and $\bm{b}^{\bm{\phi}}$ can likewise be regarded as i.i.d. Therefore, each bit in $\bm{b}^{\bm{z\phi}}$ can be modeled as an independent Bernoulli random variable, and the total number of ones in $\bm{b}^{\bm{z\phi}}$ follows a Binomial distribution. Let the Bernoulli parameter be denoted by $P_b$. Then, the probability that $n$ bits in $\bm{b}^{\bm{z\phi}}$ are equal to $\bm{1}$ is given by the binomial probability:
		\begin{equation}\label{Eq:Pn}
			\begin{split}
				{P_n}=	\big(\begin{smallmatrix}
					L-1\\
					n
				\end{smallmatrix}\big)P_b^n (1-P_b)^{L-1-n},
			\end{split}
		\end{equation}  
		where $\big(\begin{smallmatrix}
			L-1\\
			n
		\end{smallmatrix}\big)$ is the number of all possible sequences with $n$ bits being $\bm 1$, and $P_b^n (1 - P_b)^{L - 1 - n}$ is the probability of each possible sequence.
		Therefore, Eve's success probability is the sum of all probability of $N$ bit sequences tried by Eve being the actual key as 
		\begin{equation}\label{Eq:Generate}
			\begin{split}
				P_{\text{MDLG}} &= {\sum\nolimits_{n=0}^{n_{\max}}\begin{pmatrix}
						L-1\\
						n
					\end{pmatrix}P_b^n (1-P_b)^{L-1-n}} \\
			\end{split}
		\end{equation}
		According to the definition of the regularized incomplete beta function $I_x(a,b)$ \cite{wallace2010automatic}, we rewrite \eqref{Eq:Generate} as 
		\begin{equation}\label{Eq: MDLGforP}
			P_{\text{MDLG}} 
			= {I_{P_b}(L-1-n_{\max}, n_{\max}+1)}, 	
		\end{equation}
		where 
		\begin{equation}
			\begin{split}
				&\!\!\!\!\!\!{\sum\nolimits_{n=0}^{n_{\max}}\begin{pmatrix}
						L-1\\
						n
					\end{pmatrix}P_b^n (1-P_b)^{L-1-n}} \\
				&=\;\;{I_{P_b}(L-1-n_{\max}, n_{\max}+1)}. 
			\end{split}
		\end{equation}
		Under a given $n_{\max}$, Eve can verify up to $2\sum_{n=0}^{n_{\max}} \binom{L-1}{n}$ key candidates, as one candidate for $\bm{b}^{\bm{\phi}}$ corresponds to two key candidates. 
		Given Eve's capability $N$, then we have $n_{\max}=\max\{n' | 2\sum_{n=0}^{n'}\binom{L-1}{n} \leq N\}$.
		From a statistical perspective, the probability $P_b$ that a bit in $\bm{b}^{\bm{z\phi}}$ equals $\bm{1}$ matches the type of probability discussed in \cite{21qzx-tifs}, that is, the probability that the quantized bits of adjacent subchannel phase responses differ. This so-called transition probability equals $\left(\frac{1 - \rho}{2}\right)$, where $\rho$ is correlation coefficient \cite{lindqvist1978note}. 
		Finally, we can rewrite $\eqref{Eq: MDLGforP}$ to get $\eqref{Eq: MDLGforRho}$. ~\hfill$\Box$
		
		\begin{figure}[t!]
			\centering
			\begin{minipage}{0.485\columnwidth}
				\centering
				\includegraphics[width=\textwidth]{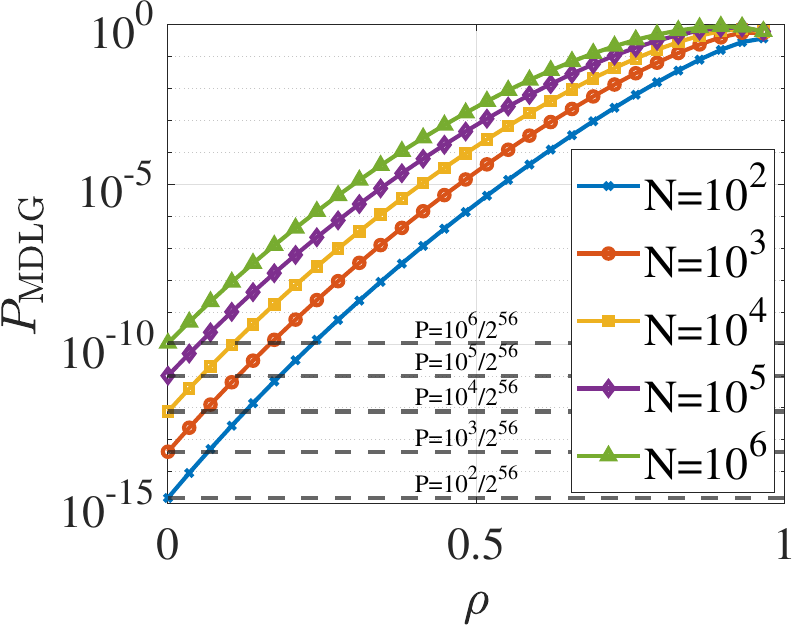}
				\subcaption[first]{Under varying $N$.}
				\label{Fig:MDLGforRho}
			\end{minipage}
			\hfill
			\begin{minipage}{0.5\columnwidth}
				\centering
				\vspace{3pt}
				\raisebox{0.02\height}{\includegraphics[width=\textwidth]{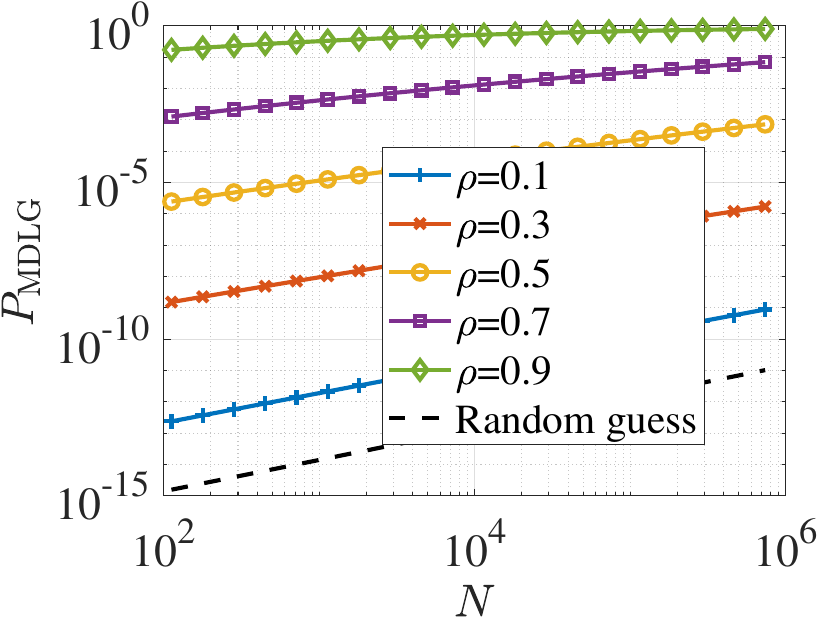}}
				\subcaption[second]{Under varying $\rho$.}
				\label{Fig:MDLGforN}
			\end{minipage}
			\caption{Attack success probabilities of MDLG.}
			\label{Fig:MDLGforPN}
		\end{figure}
		
		We use numerical simulations to illustrate the performance of MDLG. We set the key length and the number of subcarriers to be $S=L=56$. Fig.~\ref{Fig:MDLGforRho} shows Eve's success probability using MDLG as a function of channel correlation coefficient $\rho$ under Eve's capability $N=10^2$, $10^3$, $10^4$, and $10^5$. We can notice that a larger $\rho$ would result in a higher attack success probability. If Eve uses random guessing and guesses the key $N= 10^2$, $10^3$, $10^4$, and $10^5$ times, her success probabilities are $10^2/2^{56}$, $10^3/2^{56}$, $10^4/2^{56}$, $10^5/2^{56}$, and $10^6/2^{56}$, respectively, which are also shown in Fig.~\ref{Fig:MDLGforRho} as dashed lines. We can observe that when $\rho$ is 0 (i.e., the OFDM subcarriers are uncorrelated), MDLG has the same performance as random guessing. However, when the subcarriers become even weakly correlated, MDLG significantly outperforms random guessing. 
		
		Fig.~\ref{Fig:MDLGforN} shows Eve's success probability as a function of Eve's capability $N$. It is clear that a greater $N$ value leads to a higher attack success probability. In practice, there is always a limit for $N$ due to her computational limit. Overall, Fig.~\ref{Fig:MDLGforPN} shows that MDLG is a powerful attack, which leverages the imperfect wireless channel to attack PLA.

		\subsection{Multiple-Bit Mapping Case}
		We next consider the attack scenario under a multiple-bit mapping $M$, where each sub-key contains $m$ bits and the full key $\bm{s}$ of length $S$ bits is divided into $S/m$ sub-keys. In this context, $M$ is a $2^m$-ary mapping function that maps each sub-key to a phase value that is an integer multiple of $\frac{2\pi}{2^m}$.
		
		In this case, Eve's approach in finding the secret key can be extended from the binary mapping scenario. We call the strategy $\mathbf{m}$-MDLG. Specifically, Eve first obtains $\bm \Delta^{\bm z}$ based on $\eqref{Eq:Deltaz}$. Then, she replaces each element in $\bm{\Delta}^{\bm{z}}$ with its nearest multiple of $\frac{2\pi}{2^m}$ to  get a new differential sequence $\bm \Delta^{\bm z'}$, which serves as the reference for generating candidates for $\bm{\Delta}^{\bm{\phi}}$. 
		Since there are $2^m$ possible values for the first element of $\phi$ in \eqref{Eq:ppp} under multiple-bit mapping, based on $\eqref{Eq:Getphi}$ and $\eqref{Eq:ReverseMapping}$, there are $2^m$ key candidates corresponding to one candidate for $\bm{\Delta}^{\bm{\phi}}$ (in contrast to $2$ key candidates under binary mapping). She uses $\bm \Delta^{\bm z'}$ as the first candidate for $\bm{\Delta}^{\bm{\phi}}$ and generates $2^m$ key candidates accordingly for verification. If the verification fails, Eve proceeds to change one element in $\bm{\Delta}^{\bm{z'}}$ to try other key candidates. If all key candidates under one-element change still fail the verification, Eve then moves on to change $2$ and more elements in $\bm{\Delta}^{\bm{z'}}$ until the true key is identified or she reaches her capability $N$. {\hlb Since both $\bm \Delta^{\bm z'}$ and $\bm{\Delta}^{\bm{\phi}}$ are determined by the phase differences between adjacent subchannels of $\bm z$ and $\bm \phi$ and are independent of their absolute phase values, different quantization thresholds for the key mapping function $M(.)$ will not affect the attack performance, just as in the MDLG for the binary mapping case.} In the following, we present the performance of $m$-MDLG. 
		
		\begin{theorem}\label{The:mMDLG}
			Given Eve's capability $N$ and the correlation coefficient among OFDM subchannels $\rho$, the success probability of $m$-MDLG can be written as
			\begin{equation}\label{Eq: mMDLGforRho}
				P_{\text{m-MDLG}} 
				={I_{\left(1-\frac{1 + \rho}{2^{m}}\right)}\left(\frac{S}{m}-1-n_{\max}, n_{\max}+1\right)},
			\end{equation}
		\end{theorem}
		where $n_{\max}=\max\{ n' | 2^m\sum_{n=0}^{n'}\binom{\frac{S}{m}-1}{n}{(2^m-1)}^n \leq N \}$, $I_x(a,b) = \frac{B(x;a,b)}{B(1;a,b)}$ represents the regularized incomplete beta function with incomplete beta function $B(x;a,b) = \int_0^x t^{a-1}(1-t)^{b-1} \, dt$ and complete beta function $B(a,b) = B(1;a,b)$.
		
		\noindent\textit{Proof:}
		Under the $m$-bit mapping $M$, the length of $\bm{\Delta}^{\bm{z'}}$ or $\bm{\Delta}^{\bm{\phi}}$ is $\frac{S}{m} - 1$, where $S$ denotes the length of the key $\bm{s}$.
		Based on this attack, Eve incrementally changes one or more elements in $\bm \Delta^{\bm z'}$ to generate candidates for $\bm{\Delta}^{\bm{\phi}}$ and subsequently verifies the corresponding key candidates.
		When Eve decides to change $n$ elements in $\bm{\Delta}^{\bm{z'}}$, there are $\binom{\frac{S}{m} - 1}{n}$ possible combinations of element positions. 
		Since each element selected for changing has $(2^m-1)$ phase options, the total number of candidates for $\bm \Delta^{\bm \phi}$ is $\binom{\frac{S}{m} - 1}{n}{(2^m-1)}^n$. Moreover, since each candidate for $\bm{\Delta}^{\bm{\phi}}$ corresponds to $2^m$ key candidates, the total number of key candidates for a given $n$ is $2^m \binom{\frac{S}{m} - 1}{n} (2^m - 1)^n$.
		Let $n_{\max}$ denote the maximum value that $n$ can reach under Eve's capability $N$, then we have $n_{\max}=\max\{ n' | 2^m\sum_{n=0}^{n'}\binom{\frac{S}{m}-1}{n}{(2^m-1)}^n \leq N \}$. If the number of differing elements between $\bm{\Delta}^{\bm{z'}}$ and $\bm{\Delta}^{\bm{\phi}}$ does not exceed $n_{\max}$,  Eve can succeed in finding the true key $\bm s$.
		Similar to the $1$-bit mapping case, each corresponding pair of elements in $\bm{\Delta}^{\bm{z'}}$ and $\bm{\Delta}^{\bm{\phi}}$ differs independently with identical probability. Let this probability be $P_{bm}$. As a result, the probability that $\bm{\Delta}^{\bm{\phi}}$ and $\bm{\Delta}^{\bm{z'}}$ differ in $n$ elements is ${\binom{\frac{S}{m}-1}{n}}{\left(1-P_{bm}\right)^{\frac{S}{m}-n-1}}{\left(P_{bm}\right)}^{n}$. 
		Therefore, Eve's success probability is 
		\begin{equation}\label{Eq:mMDLGPrimary}
			P_{\text{m-MDLG}} 
			= \sum_{n=0}^{n_{\max}}\binom{\frac{S}{m}-1}{n}{\left(1-P_{bm}\right)^{\frac{S}{m}-n-1}}{\left(P_{bm}\right)}^{n}.
		\end{equation}
		As the value of $m$ increases, the key-phase mapping in \eqref{Eq:MappingM} becomes more fine-grained (i.e., with smaller phase intervals in $\bm{\Delta}^{\bm{z'}}$ and $\bm{\Delta}^{\bm{\phi}}$), making it more likely that the corresponding elements of $\bm{\Delta}^{\bm{z'}}$ and $\bm{\Delta}^{\bm{\phi}}$ are different and increasing the probability $P_{bm}$. When $m = 1$, the probability that the corresponding elements of $\bm{\Delta}^{\bm{z'}}$ and $\bm{\Delta}^{\bm{\phi}}$ are equal is $1-\left(\frac{1- \rho}{2}\right)=\left(\frac{1 + \rho}{2}\right)$. Based on this, for an arbitrary value of $m$, the probability that a pair of corresponding elements are equal is given by $\left(\frac{1 + \rho}{2}\right)\frac{1}{2^{m-1}} = \frac{1 + \rho}{2^m}$, and thus $P_{bm} = 1 - \frac{1 + \rho}{2^m}$. Finally, \eqref{Eq:mMDLGPrimary} can be rewritten as 
		\begin{equation}
			P_{\text{m-MDLG}} 
			={I_{\left(1-\frac{1 + \rho}{2^{m}}\right)}\left(\frac{S}{m}-1-n_{\max}, n_{\max}+1\right)}.
		\end{equation}
		~\hfill$\Box$

		We can easily verify that Theorem~\ref{The:MDLG} is a special case of Theorem~\ref{The:mMDLG} by setting $m=1$ (and thus $L=S$). Fig.~\ref{Fig:Multi} shows the impacts of $m$ and $\rho$ on the success probability of $m$-MDLG using numerical simulations ($S= 64$ and $N=10,000$). It is observed from the figure that the success probability increases as $\bm m$ increases for a fixed value of $\rho$, indicating that the binary mapping (i.e., $m=1$) is able to minimize the attack success probability and should be considered as a primary mapping scheme for Alice and Bob.

		\begin{figure}[t!]
			\centering
			\includegraphics[width=0.56\linewidth]{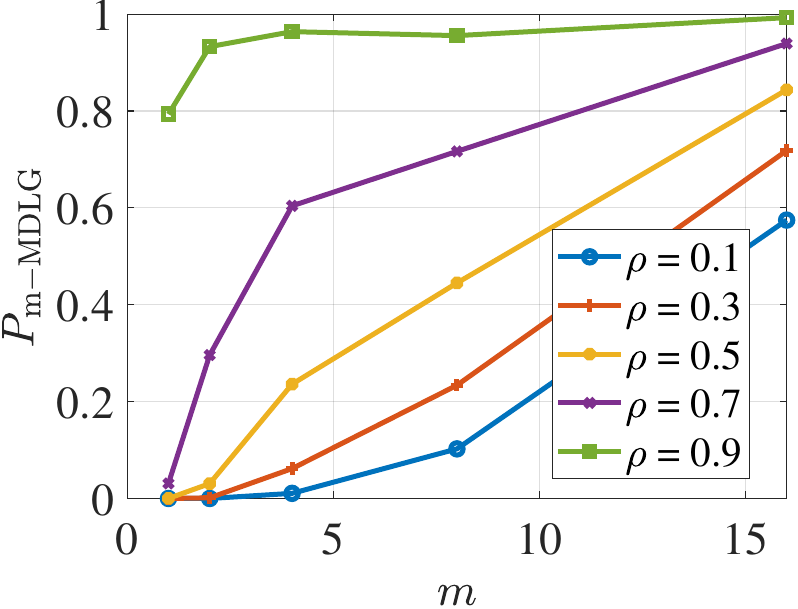}
			\caption{Attack success probabilities of $m$-MDLG.}
			\label{Fig:Multi}
		\end{figure}
		
		\section{Measurable Guideline for Security}\label{Sec:Defense Method}
		In this section, we will create a measurable guideline to determine when PLA should be used between Alice and Bob in the presence of Eve with MDLG. 
		\subsection{Design Intuition}
		From Theorem~\ref{The:MDLG}, we know that Eve can use MDLG to take advantage of imperfect channel randomness to have an effective attack to eavesdrop Alice and Bob's key during their PLA process. Thus, when the channel is not sufficiently random, Alice and Bob should not use PLA and switch back to a traditional cryptography-based authentication method. 
		
		How can Alice and Bob decide whether the channel is random enough to support their security requirement? Randomness testing, such as the NIST randomness test suite \cite{nisttest}, is a commonly used tool to test the randomness of a binary sequence with a well-studied procedure and performance even in wireless context \cite{21qzx-tifs}. If we quantify the channel response $\bm h$ in \eqref{Eq:AliBobChannel} under the mapping function $M$ into a binary sequence, we can in fact use randomness testing to decide whether the channel is random or not. In particular, a randomness test will compute a statistic, called P-value, and compare it with a threshold $\alpha$ to determine whether the sequence is random (i.e., P-value $ > \alpha$) or not (i.e., P-value $ \leq \alpha$). Changing the value of $\alpha$ will lead to a stricter or looser test \cite{21qzx-tifs}.

		\subsection{Inserting Testing into Authentication Process} 
		The randomness testing needs to be inserted into the PLA process in Fig.~\ref{Fig:basickeybased}. It involves four steps: 1) Alice sends Bob a challenge signal $S_A$ to initialize the authentication; 2) $S_A$ goes through the channel between Alice and Bob and is received by Bob as $R_B$. Meanwhile, Bob can estimate the channel response $\bm h$ in \eqref{Eq:AliBobChannel}; 3) Bob quantifies $\bm h$ into a binary sequence and conducts randomness testing on it. 4) Based on the testing result, Bob can decide whether he will continue to send his physical-layer response signal or a traditional cryptography-based signal.

		\subsection{Maximizing Efficiency while Ensuring Security}
		On one hand, increasing the threshold $\alpha$ will lead to a stricter test. This means that the test will reject more sequences as random and the authentication will switch to traditional ones, making PLA less useful. On the other hand, reducing $\alpha$ leads to a looser test, which can make PLA vulnerable under Eve's MDLG strategy. As a result, we need to create a guideline to maximize the efficiency of using PLA, while ensuring security. 
		
		We first look at how Eve's MDLG can be considered as a successful attack. There are two requirements: 1) Alice and Bob must pass the randomness testing to decide to use PLA and 2) Eve's MDLG successfully finds their key. As a result, Eve's success probability is written as  
		\begin{equation}
			\begin{split}
				P(\text{Eve succeeds}) 
				= P(T~\text{Accept}~H_0 ) P_{\text{MDLG}}, \\
			\end{split}
		\end{equation}
		where $P(T~\text{Accept}~H_0 )$ denotes the probability that randomness testing accepts the event $H_0$ that the channel is regarded as random, {\hlb and $P_{\text{MDLG}}$ denotes the conditional success probability given that the random test is passed, i.e., $P(\text{MDLG succeeds} \mid T~\text{ Accept }~H_0)$. Since the calculation logic of the MDLG attack is invariant to the outcome of the randomness test, the formulation of $P_{\text{MDLG}}$ here remains identical to the expression in \eqref{Eq: MDLGforRho}.}
		
		To ensure the security, we define a requirement
		\begin{equation}\label{Eq:BenchMark}
			P(\text{Eve succeeds}) \leq P_{\text{Benchmark}},
		\end{equation}
		where $P_\text{Benchmark}$ is the maximum acceptable attack success probability in the PLA. 
		
		The randomness testing controls when we should use PLA. In order to maximize its use efficiency, we need the channel to pass randomness testing as many times as possible. At the same time, the test should also be strict enough to make sure the security requirement \eqref{Eq:BenchMark} is met, which leads to an optimization problem of finding the optimal test threshold $\alpha$ as 
		\begin{align} \label{Eq:FinalGuideline}
			&\argmax\nolimits_{\alpha}~~~~~P(T~\text{Accept}~H_0 )\\   
			&\text{s.t.}~~~~~~~	P(\text{Eve succeeds}) \leq P_{\text{Benchmark}}.
		\end{align}
		
		Given the security requirement $P_\text{benchmark}$, we can use the guideline in \eqref{Eq:FinalGuideline} to set up the randomness testing to maximize the efficiency in using PLA. {\hlb As both $P(T~\text{Accept}~H_0 )$ and $P_{\text{MDLG}}$ are independent of the quantization thresholds, the proposed guideline remains effective regardless of the specific threshold used for the key mapping function.}
		{\hlb\subsection{Generalizability of the Proposed Guideline}
			\subsubsection{Extension to Diverse Wireless Architectures}
			While the proposed guideline is demonstrated within an OFDM framework, it can be effectively generalized to scenarios beyond OFDM. This is because it fundamentally addresses a universal vulnerability in physical layer security: the inherent physical continuity and statistical correlation between adjacent physical resources (e.g., spatial correlation in MIMO antennas or temporal correlation in fast-fading channels). In fact, these correlation-induced vulnerabilities in non-OFDM systems can still be exploited by attackers utilizing MDLG-like strategies to execute attacks.
			
			Specifically, the differential between adjacent wireless physical parameters corresponds to a differential sequence (e.g., $\bm{\Delta}^{\bm{\theta}}$ in this paper). Dictated by the underlying correlation, each element in this sequence is probabilistically more likely to be close to zero. Consequently, the attacker can generate candidates for this differential sequence in descending order of likelihood, where sequences deviating more from an all-zero state possess correspondingly lower likelihoods. With these initial candidates, the attacker can perform further scenario-specific adaptations to derive subsequent target candidates (e.g., in this paper, using $\bm{b}^{\bm{z}}$ as a reference to derive $\bm{b}^{\bm{\phi}}$, or using $\bm{\Delta}^{\bm{z'}}$ to obtain candidates for $\bm{\Delta}^{\bm{\phi}}$). These derived candidates inherently preserve the descending order of likelihood, enabling the attacker to achieve a success rate significantly higher than random guessing. Therefore, strictly applying our proposed randomness testing guideline is imperative across various wireless architectures to fundamentally restrict this widespread vulnerability.
			\subsubsection{Applicability to Dynamic and Mobile Scenarios}
			Our guideline remains applicable in dynamic and mobile scenarios. From a mechanistic perspective, neither MDLG nor the proposed guideline depends on the restrictive assumption that the nodes remain stationary throughout. 
			
			Specifically, MDLG exploits the correlation in the current Alice-Bob channel phase response associated with the moment when Alice initiates the challenge in a given authentication round, while the randomness testing performed by Bob under our guideline examines the randomness of that same channel realization. In other words, in dynamic scenarios, node mobility mainly causes the channel realizations corresponding to the moments when Alice initiates the challenge to vary across authentication rounds, but it does not change the basic logic that MDLG and the guideline focus on the same object, namely, the channel condition at the moment when Alice initiates the challenge in the current authentication round. Moreover, in most practical mobility scenarios, the duration of a single authentication procedure is still typically shorter than the channel coherence time. As a result, the channel response between Alice and Bob during that authentication round can usually be approximated as unchanged. From a system-level perspective, this further supports the applicability of the proposed guideline in dynamic scenarios.}
		
		\section{Experimental Study}\label{Sec:experiment}
		In this section, we use experimental evaluations to demonstrate the attack performance of MDLG and the effectiveness of randomness testing based design guideline. We first introduce the experimental setups and then analyze the results.
		
		\subsection{Experimental Setups}
		We conduct experiments in a realistic indoor environment, as shown in Fig.~\ref{Fig:Setup}. We fix Alice at Location 0, and place Bob at locations 1-2, 3-4, 5-6, or 7-8, which represent the short-distance line of sight (S-LoS), short-distance Non-LoS (S-NLoS), long-distance LoS (L-LoS), and long-distance NLoS (L-NoS) channel conditions, respectively. We place Eve at Location 5. {\hlb Notably, the location of Eve will not affect the MDLG attack performance or the effectiveness of the proposed guideline. This is because the attack success probability is fundamentally determined by the correlation properties of the OFDM subchannel responses between the legitimate users (i.e., Alice and Bob), and the randomness testing in the guideline is conducted exclusively on these Alice-Bob channel responses.} We use commodity WiFi devices based on Atheros AR5822/AR9580 WiFi chipsets and TP-Link WDR4300 AP. On this experimental platform, we collect over 500,000 OFDM channel responses at each location with a carrier frequency of 2.4GHz, 56 subcarriers and 20MHz bandwidth by default. The software toolkit is the Atheros CSI tool \cite{xie2018precise}.
		
		\begin{figure}[t!]
			\centering
			\includegraphics[width=0.75\linewidth]{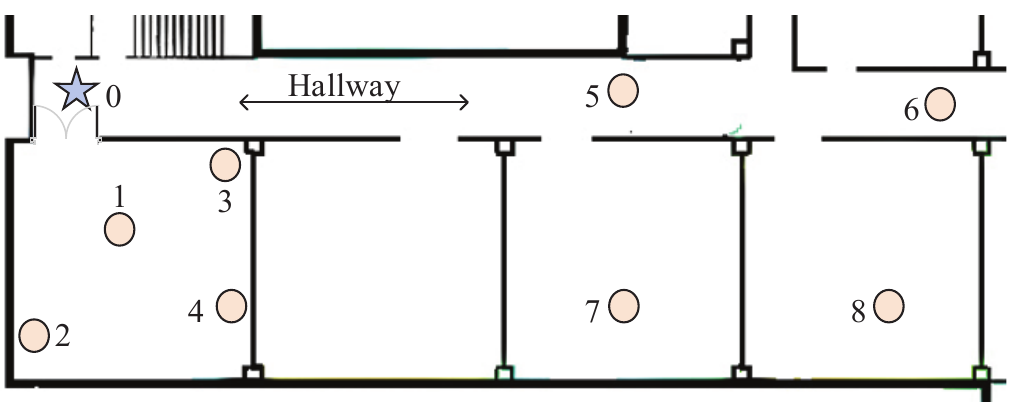}
			\caption{Experimental environment.}
			\label{Fig:Setup}
		\end{figure}

		\begin{figure*}[t!]
			\begin{minipage}{0.49\columnwidth}
				\vspace{-5pt}
				\includegraphics[width=1.05\textwidth]{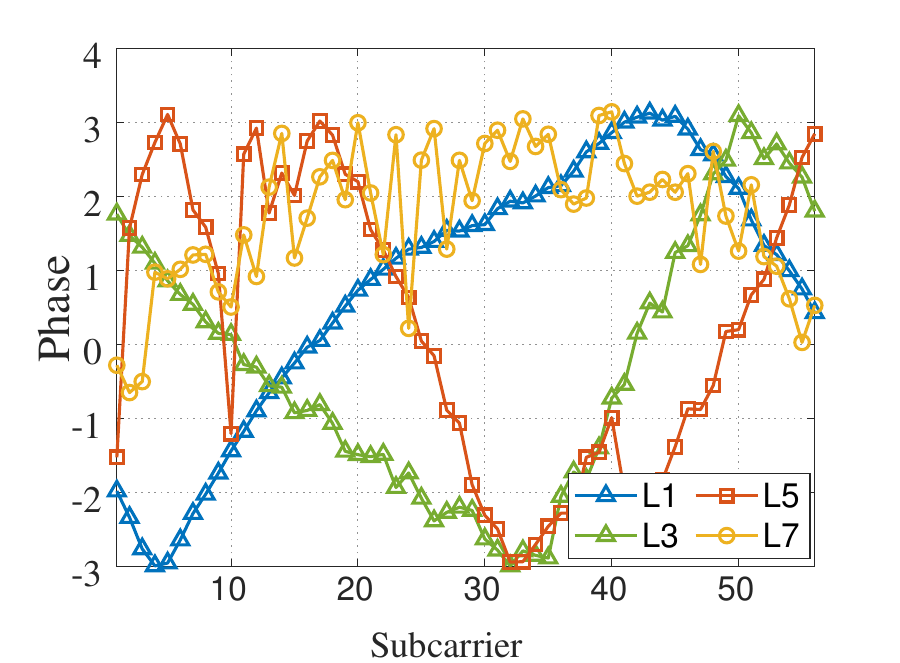}
				\caption{Channel phase responses over subchannels.}
				\label{Fig:PhaseExp}
			\end{minipage}
			\hfill
			\begin{minipage}{0.49\columnwidth}
				\centering
				\includegraphics[width=0.98\textwidth]{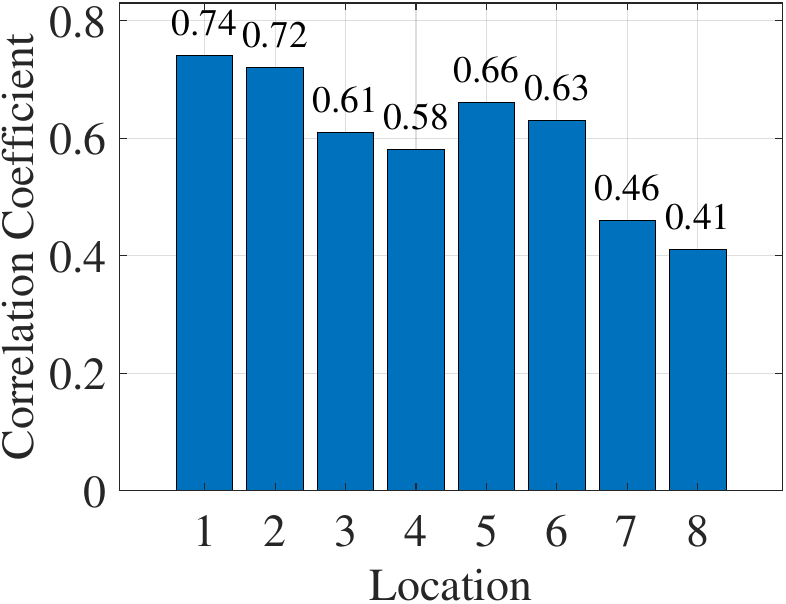}
				\caption{Correlation coefficients at different locations.}
				\label{Fig:rho_loc1-8}
			\end{minipage}
			\hfill
			\begin{minipage}{0.49\columnwidth}
				\centering
				\vspace{1pt}
				\includegraphics[width=\textwidth]{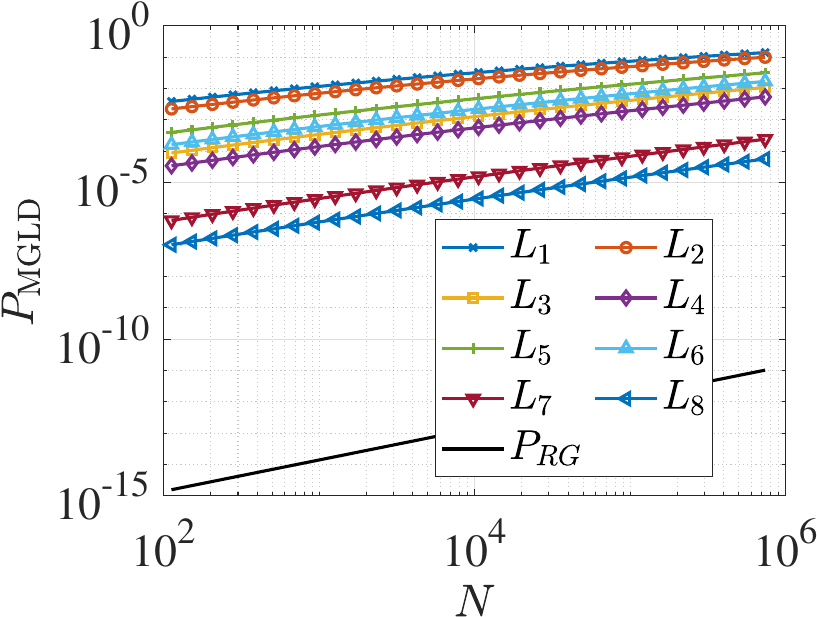}
				\caption{$P_{\text{MDLG}}$ vs $P_{RG}$ for different Eve's capabilities $N$. }
				\label{Fig:P_MDLG-and-Randomguess}
			\end{minipage}
			\hfill
			\begin{minipage}{0.49\columnwidth}
				\centering
				\vspace{0pt}
				\includegraphics[width=\textwidth]{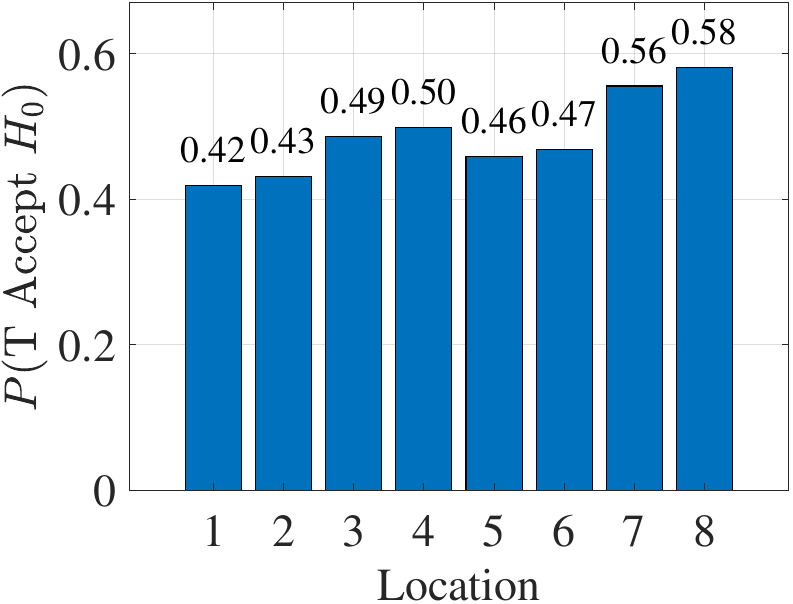}
				\caption{The test-passing probabilities with $\alpha=0.20$.}
				\label{Fig:experiment_figure_1_1}
			\end{minipage}
			\hfill
			\begin{minipage}{0.49\columnwidth}
				\centering
				\includegraphics[width=\textwidth]{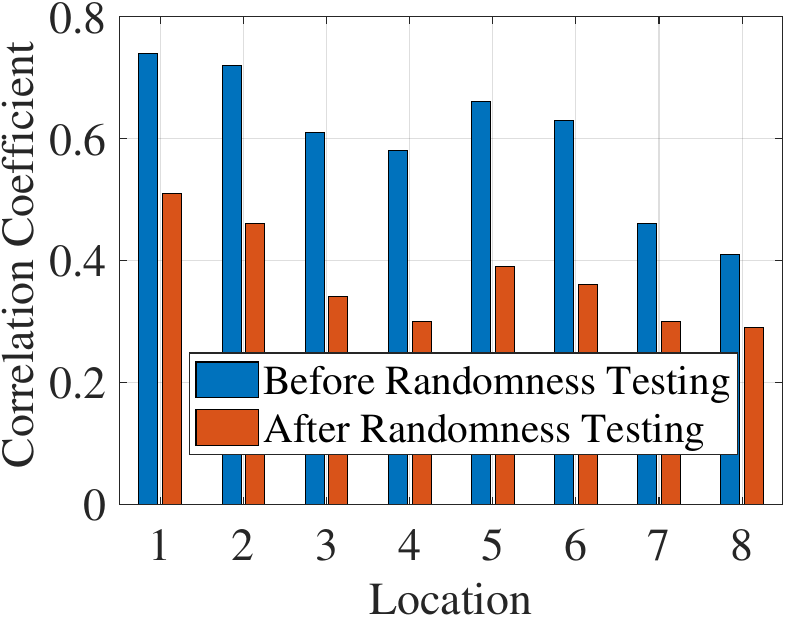}
				\caption{Correlation coefficients before and after the test $T$ ($\alpha=0.20$). }
				\label{Fig:experiment_figure_1_2}
			\end{minipage}
			\hfill
			\begin{minipage}{0.49\columnwidth}
				\centering
				\vspace{0pt}
				\includegraphics[width=\textwidth]{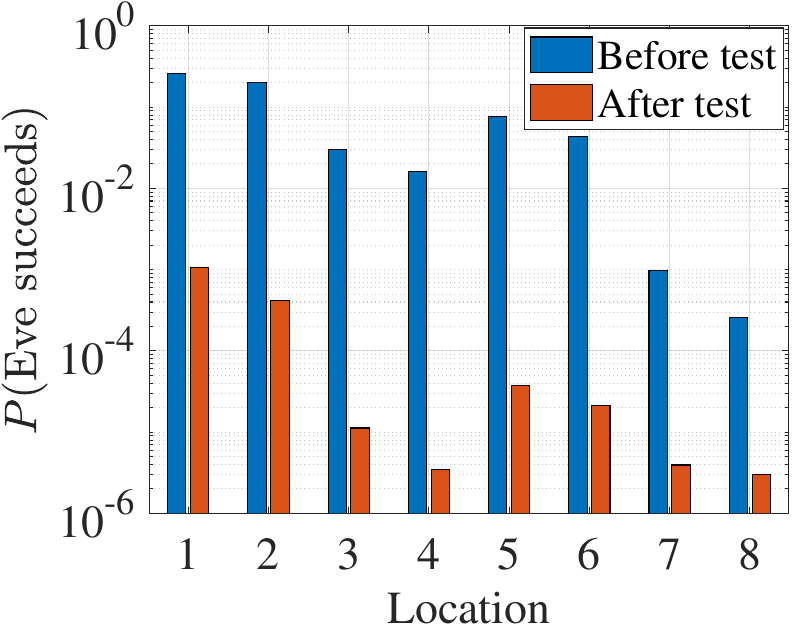}
				\caption{P(\text{Eve succeeds}) before and after the test $T$ ($\alpha=0.20$). }
				\label{Fig:experiment_figure_1_3}
			\end{minipage}
			\hfill
			\begin{minipage}{0.51\columnwidth}
				\centering
				\vspace{1pt}
				\includegraphics[width=0.98\textwidth]{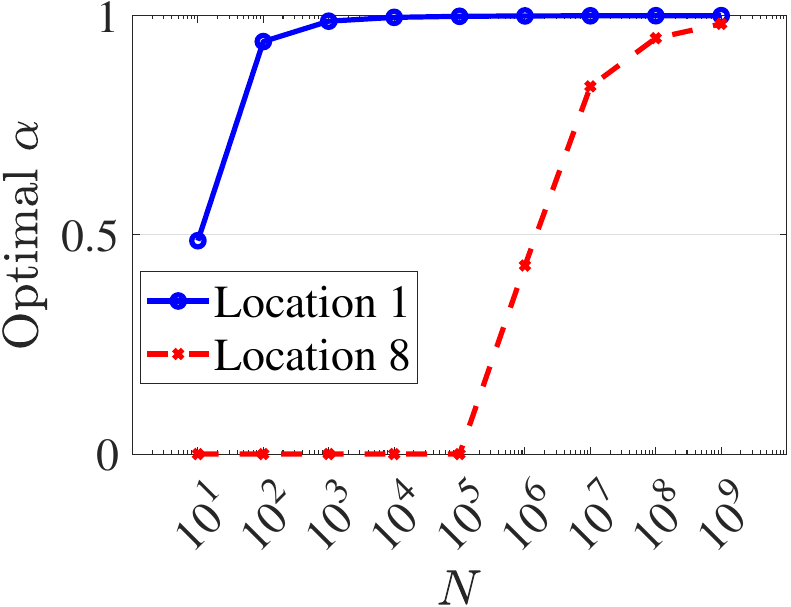}	
				\caption{The optimal $\alpha$ values under different Eve's capabilities $N$ at Locations 1 vs 8.}
				\label{Fig:experiment_figure_2_1}
			\end{minipage}
			\hfill
			\begin{minipage}{0.51\columnwidth}
				\centering
				\vspace{3pt}
				\includegraphics[width=0.98\textwidth]{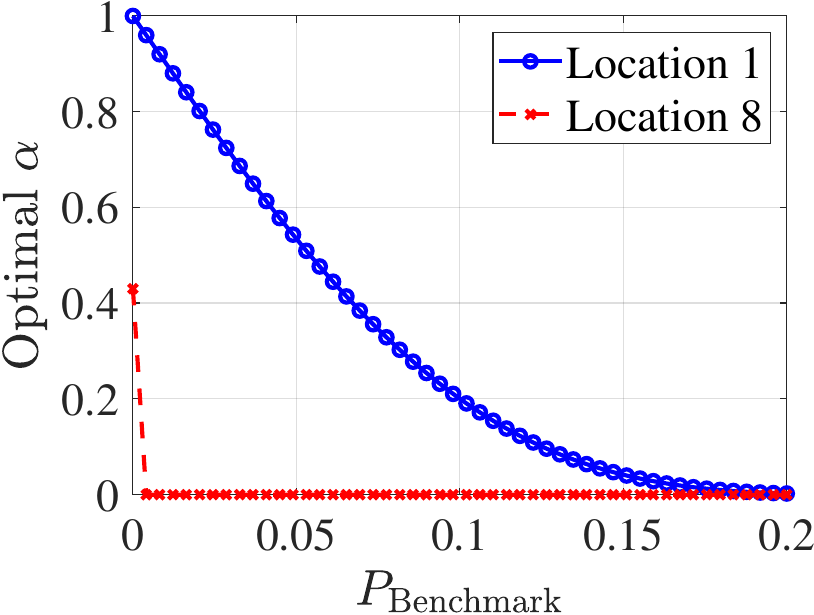}		
				\caption{The optimal $\alpha$ values under different $P_{\text{Benchmark}}$ requirements.}
				\label{Fig:experiment_figure_2_2}
			\end{minipage}
		\end{figure*}
		
		We use the frequency test in the NIST test suite \cite{nisttest} as our typical test in experiments. Alternative tests in the suite can also be used, and they are shown to yield similar testing performance \cite{jana2009effectiveness, wallace2010automatic, mathur2008radio, 21qzx-tifs}. An alternative test leads to a different mathematical expression of $P(T~\text{Accept}~H_0)$, making the optimal threshold $\alpha$ slightly different \cite{21qzx-tifs}. 
		
		The evaluation metrics in our experiments are Eve's success probability $P(\text{Eve succeeds})$ and the test-passing probability $P(T~\text{Accept}~H_0)$ (i.e., the probability that the OFDM channel passes the randomness testing), which are used to measure security and efficiency, respectively.

		\subsection{Results Analysis}
		In the following, we will thoroughly analyze the results obtained from various experiments.
		
		{\it 1) Phase Responses of OFDM Subchannels:} In this experiment, we place Bob at 8 different locations. Fig.~\ref{Fig:PhaseExp} shows an example of the channel phase responses on all OFDM subchannels. Because the neighboring locations (i.e., Locations 1-2, Locations 3-4, Locations 5-6, and Locations 7-8) have similar phase responses. We only plot the phase responses at Locations 1, 3, 5, and 7 in Fig.~\ref{Fig:PhaseExp}. The results show that the phase responses slowly vary over subchannels, confirming that the channel correlation is widespread in practice.
		
		{\it 2) Correlation Coefficient at Different Locations:} We compute the channel correlation coefficient $\rho$ for each location (i.e., placing Bob at each location and measuring Alice and Bob's channel). As shown in Fig.~\ref{Fig:rho_loc1-8}, the correlation coefficients for 8 locations are 0.74, 0.72, 0.61, 0.58, 0.66, 0.63, 0.46, and 0.41, showing that the OFDM subchannels are moderately or strongly correlated. In particular, the channel responses under the S-LoS environment show the highest correlation (e.g., $\rho \approx 0.7$ at Locations 1-2); and the L-NLoS (more scattered and multipath-rich) environment exhibits the lowest (but still moderate) correlation (e.g., $\rho \approx 0.4$ at Locations 7-8).
		
		{\it 3) Success Probability of MDLG:} In Fig.~\ref{Fig:P_MDLG-and-Randomguess}, we compute MDLG's success probability $P_\text{MDLG}$ as a function of Eve's capability $N$ at 8 different locations for Bob. The figure shows that MDLG has an order-of-magnitude advantage over random guessing to crack Alice and Bob's key during the PLA process (i.e., $P_\text{MDLG} \gg P_{RG}$ that denotes the success probability of random guessing). For example, at Locations 7 and 8, the MDLG's success probability is about $10^8$ times higher than random guessing. When the correlation among OFDM subchannels is even higher (at Locations 1 and 2), MDLG is about $10^{12}$ times better than random guessing.
		
		{\it 4) Impact of Randomness Testing:} In this experiment, we let $N$ be $10^6$ and conducted the randomness test $T$ at 8 locations and analyzed the impact of the test with a typical test threshold $\alpha=0.20$. In Fig.~\ref{Fig:experiment_figure_1_1}, we measure the test-passing probabilities range from $0.42$ to $0.58$ at Locations 1-8 with Locations 1 and 8 having the lowest and highest passing probability, respectively. We can see that with $\alpha=0.20$, the randomness test $T$ has around 50\% probability of rejecting the use of the wireless PLA.
		
		We also compute the average correlation coefficients of channel responses that actually pass the test $T$ at 8 locations (ranging from 0.30 to 0.51). We show the results in Fig.~\ref{Fig:experiment_figure_1_2} in comparison with the coefficients of all channel responses before the test $T$ (ranging from 0.41 to 0.76). It is noted from Fig.~\ref{Fig:experiment_figure_1_2} that correlation coefficients of the channel responses accepted by $T$ are smaller than those before $T$, indicating the test $T$ can help filter out the defective channel cases to combat the MDLG attack.
		
		Fig.~\ref{Fig:experiment_figure_1_3} plots Eve's success probabilities at 8 locations when Eve uses MDLG to attack the PLA before and after the deployment of test $T$. It is clear from the figure that Eve has a much lower success probability if the wireless channel is first tested by $T$ before it is used for PLA. For example, at Location~8, $P(\text{Eve succeeds}) = 2.561\times 10^{-4}$ without $T$ and it becomes $2.986\times10^{-6}$ when $T$ is used.

		{\it 5) Optimal $\alpha$ under Different Conditions:} Figs.~\ref{Fig:experiment_figure_1_1}, \ref{Fig:experiment_figure_1_2}, and Fig.~\ref{Fig:experiment_figure_1_3} show that randomness testing chosen with a typical value $\alpha=0.20$ can substantially reduce Eve's success. As designed in our guideline in \eqref{Eq:FinalGuideline}, the $\alpha$ value can be optimized to maximize the  efficiency of PLA while maintaining security. 
		
		To see how our guideline \eqref{Eq:FinalGuideline} determines the optimal $\alpha$, we consider two locations, Locations 1 and 8, which represent the S-LoS (i.e., strongly correlated) and L-LoS (i.e., weakly correlated) environments, respectively. We let the maximum allowable attack success probability $P_{\text{Benchmark}}$ be $10^{-4}$. Based on our guideline \eqref{Eq:FinalGuideline}, we compute the optimal $\alpha$ values at two locations as a function of Eve's capability $N$, shown in Fig.~\ref{Fig:experiment_figure_2_1}. The figure demonstrates that Location~1 needs a larger $\alpha$ value than Location~8 in the test $T$ to ensure security because Location 1 leads to a much stronger OFDM subchannel correlation than Location 8. We also notice from the figure that as Eve becomes more capable with a larger $N$ value, the optimal value $\alpha$ increases to 1. This means that the randomness test $T$ will consider more channel cases defective (i.e., not meeting the benchmark security requirement in \eqref{Eq:FinalGuideline}) to be used for PLA when facing a more powerful attack. 
		
		Fig.~\ref{Fig:experiment_figure_2_2} also shows the optimal $\alpha$ value as a function of the benchmark requirement $P_{\text{Benchmark}}$ under Eve's capability $N = 10^6$. Fig.~\ref{Fig:experiment_figure_2_2} illustrates that as $P_{\text{Benchmark}}$ increases, the optimal $\alpha$ value decreases (even to 0). This is because a higher $P_{\text{Benchmark}}$ reduces the security requirement (i.e., Alice and Bob become more tolerant of an attack), allowing for a lower testing threshold to enhance the efficiency of using PLA.

		\section{Related Work}\label{Sec:Relatework}
		In this section, we present the research related to our work. 
		
		\vspace{0.01cm}
		{\noindent\bf Improving PLA security:} PLA performance heavily relies on the random wireless channel. Some studies focused on making the PLA procedure appear more random to the adversary \cite{wu2016artificial, Poor2017security}. In \cite{wu2016artificial}, Tikhonov-distributed artificial noise was introduced to interfere with the key-related signal for resisting potential key recovery attacks. This method not only affects the reception of the eavesdropper but also that of the intended user as signals are both degraded. Due to the presence of noise, the decryption process of the secret key will become more complex. In \cite{Poor2017security}, instead of using all available resources for message transmission, a certain part of them were used for randomization by adding ``dummy'' messages unknown to Eve such that Eve would be saturated with useless information. These studies focused on empirically improving PLA schemes but lacked a scientific guideline to assess whether the enhanced system meets security standards. Our work offers a well-defined guideline via detailed mathematical modeling to determine when PLA is sufficiently secure.
		
		\vspace{0.01cm}
		{\noindent\bf Using randomness testing in wireless security:} Randomness testing has been proposed to test the key sequence used for authentication \cite{21qzx-tifs, Upadhyay2018randomness, mathur2008radio, jana2009effectiveness,soto2000,Manucom2017,hring2015}. Although they share the same objective, the methods for generating key sequences and the randomness tests used can differ. For example, \cite{Upadhyay2018randomness} establishes the key sequence for Alice and Bob by extracting and quantifying their Received Signal Strength Indicator (RSSI). The work in \cite{hring2015} proposed a reduced set of statistical tests that have low complexity and therefore can be executed at the runtime even on severely resource-constrained devices. The work in \cite{21qzx-tifs} investigated how to configure randomness testing to meet the security strength for creating a key from the wireless channel. Compared to these studies, our work focuses on the authentication process rather than the key establishment process. Specifically, our design uses randomness testing to determine when PLA is secure, rather than just establishing a secret key or evaluating its strength.

		\vspace{0.01cm}
		{\noindent\bf Using other physical layer features for security:} Existing studies have explored authentication mechanisms using other physical layer features. For example, \cite{hou2014physical, hua2018accurate} utilized the combination of device-dependent biases in radio frequency and the mobility-induced Doppler shift, characterized as a time-varying carrier frequency offset, as a radiometric signature for authentication. The work in \cite{shen2022towards} leveraged the deep metric learning to train a radio frequency fingerprint identification by exploiting channel randomness features and data augmentation for authentication on low-power long-range devices. These studies may be hardware-specific and usually require a training process by collecting labeled data for machine learning, which can be cumbersome to perform and adjust over time in practice. Our work is complementary to these studies and offers a measurable guideline to determine when the wireless channel is sufficient to support PLA. 
		
		\section{Conclusion}\label{Sec:Conclusion}
		In this paper, we explored the potential vulnerability of PLA under practical OFDM channel conditions. Based on the correlations among subchannels, we introduce a novel adversary model named MDLG, which helps an eavesdropper to effectively infer Alice and Bob’s secret key from its accessible information. To defend against such a critical attack, we created a measurable guideline that uses randomness testing to first test whether the wireless channel can indeed support the PLA process under the MDLG attack model given a security requirement. We formulated the guideline as an optimization problem to solve the optimal setup for randomness testing. We also conducted comprehensive real-world experiments to show that the guideline can efficiently protect the PLA against the proposed MDLG attack.

	\bibliographystyle{IEEEtran}
	\bibliography{WhenUseAuth}

	\begin{IEEEbiography}[{\includegraphics[width=1in,height=1.5in,clip,keepaspectratio]{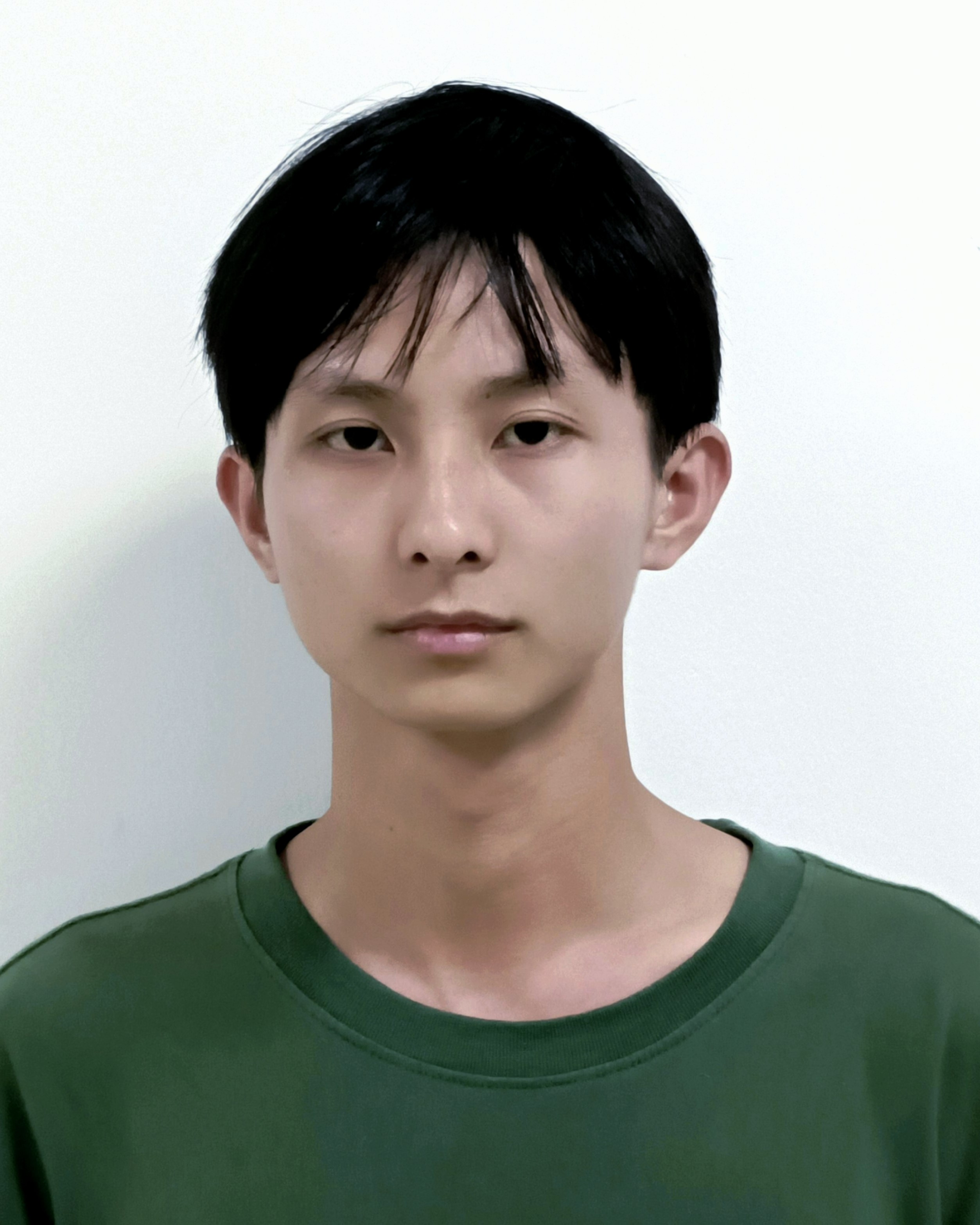}}]{Haiyun Liu} (Student Member, IEEE) received the B.S. degree in electrical engineering and automation from Shanghai University, Shanghai, China, in 2019, and the M.S. degree in signal and information processing from Sichuan University, Chengdu, China, in 2023. He is currently a Ph.D. student in the Bellini College of Artificial Intelligence, Cybersecurity and Computing, University of South Florida, Tampa, FL, USA. His research interests include security and privacy in wireless communications, and federated learning for networks.
	\end{IEEEbiography}
	\begin{IEEEbiography}[{\includegraphics[width=1in,height=1.8in,clip,keepaspectratio]{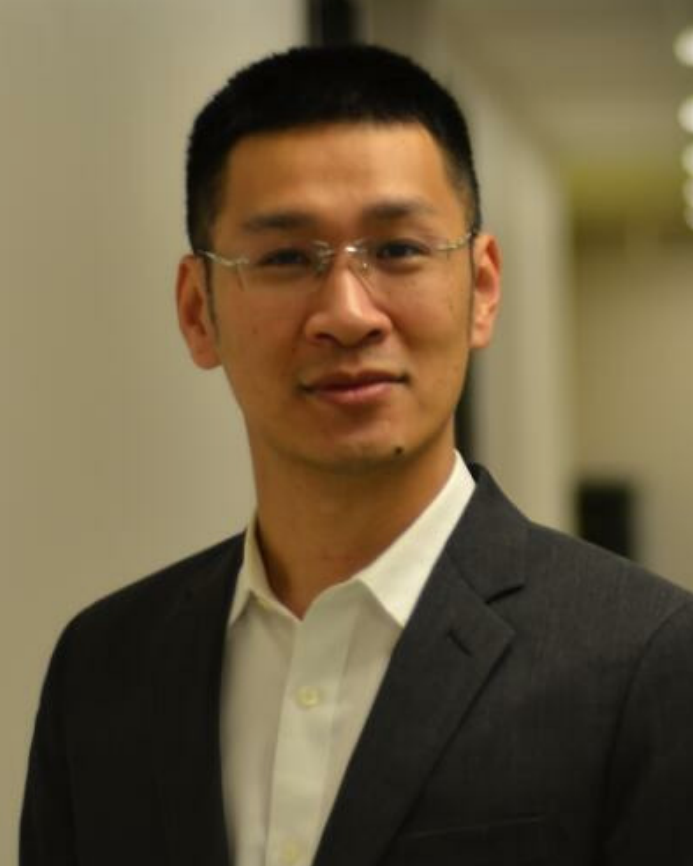}}]{Shangqing Zhao} (Member, IEEE)  received the Ph.D. degree in computer science from the University of South Florida in 2021. He is currently an Assistant Professor with the School of Computer Science, University of Oklahoma. His research interests include network and mobile system design and security.
	\end{IEEEbiography}
	\begin{IEEEbiography}[{\includegraphics[width=1in,height=1.8in,clip,keepaspectratio]{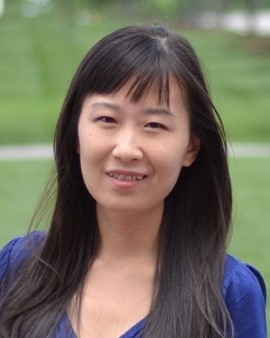}}]{Yao Liu} (Senior Member, IEEE) received the Ph.D. degree in computer science from North Carolina State University, in 2012. She is an professor in the Bellini College of Artificial Intelligence, Cybersecurity and Computing, University of South Florida. Her research interests include in the security applications for cyberphysical systems, Internet of Things, and machine learning. She was an NSF CAREER Award recipient in 2016. She also received the ACM CCS Test-of-Time Award by ACM SIGSAC in 2019.
	\end{IEEEbiography}
	\begin{IEEEbiography}[{\includegraphics[width=1in,height=1.8in,clip,keepaspectratio]{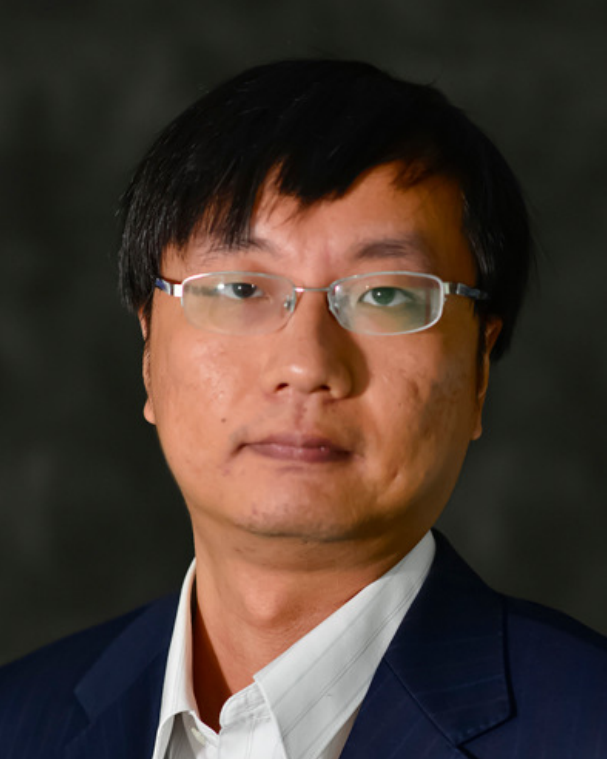}}]{Zhuo Lu} (Senior Member, IEEE) received the Ph.D. degree from North Carolina State University, in 2013. He is an associate professor with the Department of Electrical Engineering, University of South Florida. His research has been mainly focused on both theoretical and system perspectives on communication, network, and security. He received the NSF CISE CRII award in 2016, the Best Paper Award from IEEE GlobalSIP in 2019, and the NSF CAREER award in 2021.
	\end{IEEEbiography}
	
\end{document}